\documentclass[
notitlepage,
floatfix,
aps,
prl,
reprint,
twocolumn,
groupedaddress,
]{revtex4-2}
\usepackage{graphicx}
\usepackage{amsmath,amsfonts,amssymb,amscd}
\usepackage{mathtools}
\usepackage{ascmac}
\usepackage{mathrsfs}
\usepackage{amsthm}
\usepackage{dsfont}
\usepackage{enumitem}
\usepackage{mathtools}
\usepackage{braket}
\usepackage{bm}
\usepackage{verbatim}
\usepackage{framed}
\usepackage{listings}
\usepackage{tikz}
\usetikzlibrary{calc}
\usetikzlibrary{positioning}
\usetikzlibrary{arrows.meta}
\usepackage[linesnumbered,ruled,vlined]{algorithm2e}
\usepackage{booktabs}
\usepackage{multirow}

\usepackage[dvipsnames]{xcolor}
\usepackage{hyperref}
\hypersetup{
    pdfnewwindow=true,
    colorlinks=true,
    citecolor=Blue,
    linkcolor=Blue,
    urlcolor=Blue
}

\DeclareMathOperator{\supp}{supp}
\DeclareMathOperator{\Tr}{Tr}

\newcommand{\eq}[1]{Eq.~\eqref{#1}}

\newcommand{\given}{\mathrm{Given}}
\newcommand{\minimize}[1]{\underset{{#1}}{\mathrm{minimize}}}
\newcommand{\maximize}[1]{\underset{{#1}}{\mathrm{maximize}}}
\newcommand{\subto}{\text{subject to}}

\renewcommand{\ge}{\geqslant}
\renewcommand{\le}{\leqslant}

\newcommand{\ketbra}[2]{\ket{{#1}}\!\!\bra{{#2}}}

\newcommand{\cA}{\mathcal{A}}

\newcommand{\cD}{\mathcal{D}}

\newcommand{\cH}{\mathcal{H}}
\newcommand{\cI}{\mathcal{I}}

\newcommand{\cL}{\mathcal{L}}

\newcommand{\cS}{\mathcal{S}}

\newcommand{\cX}{\mathcal{X}}

\newcommand{\rA}{\mathrm{A}}
\newcommand{\rB}{\mathrm{B}}

\newcommand{\rE}{\mathrm{E}}

\newcommand{\rH}{\mathrm{H}}

\newcommand{\rS}{\mathrm{S}}

\newcommand{\sM}{\mathscr{M}}

\newcommand{\B}{\mathsf{B}}

\newcommand{\E}{\mathsf{E}}
\newcommand{\F}{\mathsf{F}}
\newcommand{\G}{\mathsf{G}}

\newcommand{\J}{\mathsf{J}}

\newcommand{\M}{\mathsf{M}}

\newcommand{\one}{\mathds{1}}
\newcommand{\zero}{\mathds{O}}

\newcommand{\mbN}{\mathbb{N}}
\newcommand{\mbF}{\mathbb{F}}

\newcommand{\mbM}{\mathbb{M}}
\newcommand{\mbE}{\mathbb{E}}

\newcommand{\guess}{\mathrm{guess}}
\newcommand{\prev}{\mathrm{prev}}
\newcommand{\new}{\mathrm{new}}

\renewcommand{\emph}[1]{{\textbf{#1}}}

\newtheorem{theorem}{Theorem}
\newtheorem{proposition}{Proposition}

\newtheorem{lemma}{Lemma}
\newtheorem{fact}{Fact}
\theoremstyle{definition}
\newtheorem{definition}{Definition}
\newtheorem{remark}{Remark}

\newtheorem{example}{Example}

\begin{document}
\title{All star-incompatible measurements can certify steering-based randomness}

\author{Shintaro Minagawa}
\email{minagawa.shintaro@gmail.com}
\affiliation{Aix-Marseille University, CNRS, LIS, Marseille, France}

\author{Ravi Kunjwal}
\email{quaintum.research@gmail.com}
\affiliation{Universit\'e libre de Bruxelles, QuIC, Brussels, Belgium}
\affiliation{Aix-Marseille University, CNRS, LIS, Marseille, France}

\begin{abstract}
Certifying that quantum randomness generated by untrusted devices is unpredictable to an attacker (say, Eve) is crucial for device-independent security. 
Bipartite protocols where only one of the parties is trusted are termed \textit{one-sided device-independent} (1SDI) or \textit{steering-based} protocols, where the untrusted party (say, Alice) performs measurements on her part of a bipartite entangled state to steer the subsystem of the trusted party (say, Bob) into different ensembles (collectively, an \textit{assemblage}) of quantum states. 
Recent work has shown that an assemblage has certified randomness if and only if its realization requires a set of measurements that are star-incompatible, \textit{i.e.}, the measurement setting of interest for the guessing probability of Eve is incompatible with at least one of the remaining measurement settings of Alice.
However, can one \textit{always} extract randomness from star-incompatible measurements? And if so, how much randomness can one extract from them? Here we demonstrate that one can extract steering-based randomness from \textit{any} set of star-incompatible measurements, thereby establishing an equivalence between star-incompatibility and steering-based randomness. 
We then provide a lower bound on the minimum guessing probability for a given set of measurements that is computable via semidefinite programming. 
This bound is tight when the Hilbert space dimension of Bob's system is at least that of Alice's.
Finally, to provide a quantitative handle on star-incompatibility and the randomness extractable from it, we introduce a weight-based measure of star-incompatibility and lower bound the amount of star-incompatibility required to certify a given amount of randomness.
\end{abstract}

\maketitle

\textit{Introduction---}.
The intrinsic unpredictability of quantum measurement outcomes \cite{born1926zur} is at the heart of quantum randomness certification protocols \cite{masanes2006general,dhara2014can,yuan2015intrinsic,yuan2019quantum,senno2023quantifying,dai2023intrinsic,meng2024maximal,curran2025maximal}, with applications in such areas as security~\cite{herrero-collantes2017quantum}.
To circumvent the possibility that random number generators may leak information to an attacker, certifying the unpredictability of the randomness generated by such devices is necessary.
Device-independent randomness certification~\cite{colbeck2011quantum,colbeck2011private,pironio2010random,pironio2013security,acin2016certified} is a method that uses the violation of a Bell inequality~\cite{bell1964einstein,brunner2014bell,scarani2019bell} from the outputs of multiple devices to certify randomness. Due to the high implementation cost of such protocols \cite{liu2018device}, a more relaxed version called one-sided device-independent (1SDI) randomness certification was proposed, where one party can be trusted~\cite{law2014quantum,passaro2015optimal,skrzypczyk2018maximal}.
This Letter aims to provide a qualitative and quantitative understanding of the measurement incompatibility required in 1SDI randomness certification.

The 1SDI setup is based on a concept that  originates from the Einstein-Podolsky-Rosen paradox~\cite{einstein1935can}, which was formalized as quantum steering~\cite{wiseman2007steering} later (for a review, see~\cite{cavalcanti2016quantum,uola2020quantum}).
Alice and Bob share a quantum state, and Alice has a device that performs several measurements in a measurement set.
The family of post-measurement states on Bob's side, conditioned on Alice's choice of a measurement setting and the measurement outcome, is known as an \textit{assemblage}.
The shared quantum state, however, may be correlated with an attacker called Eve, who has access to a quantum memory and aims to guess Alice’s outcomes. Ref.~\cite{passaro2015optimal} presented an approach to calculate Eve's guessing probability for a fixed setting via the  assemblage prepared on Bob's system.
The assemblage exhibits certified randomness whenever Eve's optimal guessing probability is strictly lower than $1$.

Recently, Ref.~\cite{li2025necessary} fully characterized the structure of assemblages with no certified randomness.
They showed that an assemblage has no certified randomness if and only if it is realizable by a set of \textit{star-compatible} (``$K_{1,m-1}(x^*)$-compatible" in \cite{li2025necessary}) measurements on Alice's side, \textit{i.e.}, the compatibility structure of Alice's set of measurements is such that the measurement targeted by Eve's guess is compatible with all other measurements in the set.
In other words, generating an assemblage that exhibits certified randomness requires a star-incompatible set of measurements. 
This raises some natural questions: What kind of star-incompatible set of measurements can generate an assemblage possessing certified randomness? Also, how can we identify star-incompatible measurements and a state that can prepare assemblages with high certified randomness?

We demonstrate that all star-incompatible sets of measurements can generate steering-based randomness.
This establishes an equivalence between star-incompatibility and certified randomness in the 1SDI setup.
Beyond this qualitative equivalence, we also provide a lower bound of Eve's guessing probability for a given set of measurements, which is computable by semidefinite programming.
Notably, this bound is exact when Bob's system is of the same or greater Hilbert space dimension than Alice's system.
In addition, we provide a lower bound on the amount of star-incompatibility of Alice's measurements required to generate an assemblage with a certain level of certified randomness in the 1SDI setup.
We introduce \textit{star-incompatibility weight} as a measure of star-incompatibility, adapting the weight-based measure originally introduced for  incompatibility~\cite{pusey2015verifying}.
In quantitatively combining star-incompatibility and certified randomness from a resource-theoretic perspective~\cite{chitambar2019quantum}, we provide tools for one-sided estimation of star-incompatibility through randomness.

\textit{Preliminaries---}.
A set of linear operators from a complex Hilbert space $\cH$ to itself is denoted as $\cL(\cH)$.
The subset $\cL_\rH(\cH)\subset\cL(\cH)$ denotes all Hermitian operators.
The subset $\cL_\mathrm{PSD}(\cH)\subset\cL_\rH(\cH)$ denotes the set of all positive semidefinite operators, and for $X\in\cL_\mathrm{PSD}(\cH)$, we write $X\succeq\zero$ where $\zero$ is the zero operator.
We denote the identity operator by $\one$.

We introduce some notation based on the standard formalism of quantum theory following Refs.~\cite{nielsen2010quantum,wilde2017quantum,watrous2018theory}.
A quantum state of a quantum system $\rA$ associated with a complex Hilbert space $\cH_\rA$, whose dimension $\dim\cH_\rA$ is finite, is described by a positive semidefinite operator $\rho_\rA\in\cL_\mathrm{PSD}(\cH_\rA)$ that satisfies $\Tr\rho_\rA=1$.
Let $\cS(\cH_\rA)$ denote the set of all quantum states on system $\rA$.
We denote POVM on $\cH_\rA$ as a family of positive semidefinite operators $\M^\cA_\rA:=\{M^a_\rA\}_{a\in\cA}$, where $\cA$ is the set of outcomes, and $M^a_\rA\in\cL_{\mathrm{PSD}}(\cH_\rA)$ satisfies $\sum_{a\in\cA}M^a_\rA=\one_\rA$ hold.
Let $\sM^\cA_\rA$ denote the set of all POVMs on $\cH_\rA$ with the outcome set $\cA$.
We calculate the probability of getting an outcome $a\in\cA$ by using a formula $\Tr[M^a_\rA\rho_\rA]$.

We also consider a set of POVMs on the same Hilbert space $\cH_\rA$, $\mbM^{\cA|\cX}_\rA:=\{\M^{\cA|x}_\rA\}_{x\in\cX}$, where $\cX$ is the set of labels of POVMs, the number of elements of $\cX$ is $n:=|\cX|$, and $\M^{\cA|x}_\rA:=\{M^{a|x}_\rA\}_{a\in\cA}\in\sM^\cA_\rA$ is the $x$-th POVM.\footnote{For each setting $x$, the set of outcomes may differ. However, if necessary, we can add the zero operator $\zero_\rA$ appropriately so that the sets of outcomes are the same across all $x\in\cX$. Consequently, the set of outcomes is denoted as $\cA$, independent of $x$.} 
We use the \textit{programmable measurement device}~\cite{buscemi2020complete}(PMD) terminology as a concise way to refer to such an indexed family of POVMs as a single measurement device.

When quantum systems A and B are associated with complex Hilbert spaces $\cH_A$ and $\cH_B$, respectively, their composite system $\rA\rB$ is associated with a tensor product Hilbert space $\cH_\rA\otimes\cH_\rB$.
Let Alice and Bob have quantum systems $\rA$ and $\rB$, and share a quantum state $\rho_{\rA\rB}\in\cS(\cH_\rA\otimes\cH_\rB)$, and let Alice choose a measurement $\M^{\cA|x}_\rA=\{M^{a|x}_\rA\}_{a\in\cA}$ from a PMD $\mbM^{\cA|\cX}_\rA$ and perform it to get an outcome $a$.
Then, Bob's unnormalized post-measurement state becomes
\begin{equation}\label{eq:assemblage}
    \sigma_{\rB }^{a|x}=\Tr_\rA\left[\left(M^{a|x}_\rA\otimes\one_\rB\right)\rho_{\rA\rB}\right].
\end{equation}
The set of all Bob's unnormalized post-measurement states, $\sigma^{\cA|\cX} _\rB :=\{\sigma^{a|x} _\rB \}_{a\in\cA,x\in\cX}$, is called an \textit{assemblage}.
Let us denote an assemblage that can be realized by Eq.~\eqref{eq:assemblage} as $\sigma(\mbM^{\cA|\cX}_\rA,\rho_{\rA\rB})$.

\begin{figure}[t]
\centering
\includegraphics[width=0.8\linewidth]{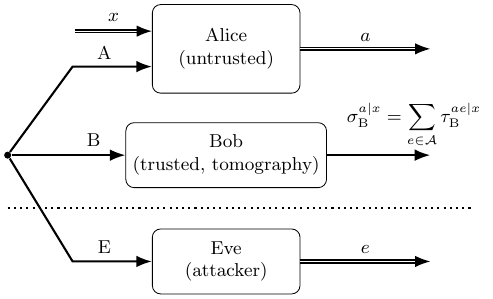}
\caption{One-sided device-independent (1SDI) randomness certification. 
The single arrows represent quantum information, and the double arrows represent classical information.
Alice chooses the $x$-th measurement from her set of measurements, and gets an outcome $a\in\cA$. Bob is in a trusted lab, and can fully characterize the post-measurement state by state tomography ~\cite{paris2004quantum}, obtaining an assemblage $\{\sigma^{a|x}_\rB\}_{a\in\cA,x\in\cX}$. An attacker named Eve, space-like separated from Alice and Bob, guesses Alice's measurement outcomes. The guess is successful if $e=a$.}
\label{fig:setup}
\end{figure}

\textit{Equivalence between star-incompatibility and certified randomness in the 1SDI setup}---.
We explain the 1SDI randomness certification following Ref.~\cite{passaro2015optimal}.
Suppose that Alice, Bob, and Eve have quantum systems $\rA,\rB$, and $\rE$ associated with a finite-dimensional complex Hilbert space $\cH_\rA,\cH_\rB$ and $\cH_\rE$, respectively.
They share a quantum state $\rho_{\rA\rB\rE}\in\cS(\cH_\rA\otimes\cH_\rB\otimes\cH_\rE)$, and Alice performs the $x$-th measurement $\M^{\cA|x}_\rA$ in a PMD $\mbM^{\cA|\cX}_\rA$ and gets an outcome $a\in\cA$.
Eve, an attacker separated from Alice and Bob, performs a POVM $\E^\cA_\rE:=\{E^e_\rE\}_{e\in\cA}$ to guess Alice's outcomes.
Let us introduce $\tau^{ae|x}_\rB=\Tr_{\rA\rE}[(M^{a|x}_\rA\otimes\one_\rB\otimes E^e_\rE)\rho_{\rA\rB\rE}]$.
When Bob obtains an unnormalized state $\sigma^{a|x}_\rB$, it holds that $\sum_{e\in\cA}\tau^{ae|x}_\rB=\sigma^{a|x}_\rB$ because Eve's interventions cannot be detected by Bob.
Eve's average successful guessing probability for the $x^*$-th setting is $\sum_{e\in\cA}\Tr\tau^{ee|x^*}_\rB$.
Fig.~\ref{fig:setup} depicts this setup.

For a given assemblage $\sigma^{\cA|\cX}_\rB$, maximizing the guessing probability over all PMDs, Eve's POVMs, and the shared state is not a semidefinite program (in short, SDP~\cite{boyd2004convex,skrzypczyk2023semidefinite}) as it is, but Ref.~\cite{passaro2015optimal} reduced it to an SDP.
    \begin{align}
        p^{x^*}_{\rm guess}(\sigma^{\cA|\cX} _\rB)&=\maximize{\tau_\rB^{ae|x}\in\cL_\rH(\cH_\rB)} \sum_{e\in\cA}\Tr\tau^{ee|x^*}_\rB\notag\\
        \subto\quad &\sum_{e\in\cA}\tau_\rB^{ae|x}=\sigma^{a|x}_\rB \; (\forall x\in\cX,\forall a\in\cA),\notag\\
        &\sum_{a\in\cA}\tau_\rB^{ae|x}=\sum_{b\in\cA}\tau_\rB^{be|x^*}\;(\forall x\in\cX,\forall e\in\cA),\notag\\
        &\tau_\rB^{ae|x}\succeq\zero_\rB\;(\forall x\in\cX,\forall a,e\in\cA).\label{eq:guessing_assemblage}
    \end{align}

$p_\guess^{x^*}(\sigma^{\cA|\cX}_{\rB})<1$ means that Eve cannot perfectly predict Alice's measurement outcome, \textit{i.e.}, \textit{steering-based randomness} is certified.
Note that Eq.~\eqref{eq:guessing_assemblage} does not rely on any knowledge or assumptions about the POVM performed by Alice.
It only uses Bob's assemblage information to calculate the randomness, \textit{i.e.}, in a \textit{one-sided device-independent} (1SDI) manner.

Realizing nonlocal correlations and using them to achieve certified randomness requires a set of measurements to be incompatible~\cite{heinosaari2016invitation,guhne2023incompatible}. In the case of 1SDI protocols, Ref.~\cite{li2025necessary} introduced a certain type of incompatibility of a PMD, which we term \textit{star-incompatibility} in this paper (and which was originally called $K_{1,n-1}(x^*)$-incompatibility in \cite{li2025necessary}), and identified it as a necessary resource for certified randomness. We formally define star-(in)compatibility below:
\begin{definition}
    A PMD $\mbM^{\cA|\cX}_\rA$ is said to be star-compatible with respect to setting $x^*\in \cX$ if there exist POVMs $\{\J^{\cA\times\cA|x}_\rA\}_{x\in\cX\setminus\{x^*\}}$ where $\J^{\cA\times\cA|x}_\rA:=\{J^{ab|x}_\rA\}_{a,b\in\cA}\in\sM^{\cA\times\cA}_\rA$ such that~\footnote{This paper adopts a definition using marginalization, whereas Ref.~\cite{li2025necessary} employs a definition using post-processing. However, Ref.~\cite{li2025necessary} also rewrites the definition equivalently using standard argumentation that replaces the post-processing with a deterministic process, which effectively corresponds to taking marginalization (see e.g.,~\cite{skrzypczyk2023semidefinite}).}
    \begin{align}
        M^{a|x}_\rA&=\sum_{b\in\cA}J^{ab|x}_\rA\;(\forall x\in\cX\setminus\{x^*\},\;\forall a\in\cA),\\
        M^{b|x^*}_\rA &=\sum_{a\in\cA}J^{ab|x}_\rA\;(\forall x\in\cX\setminus\{x^*\},\;\forall b\in\cA).
    \end{align}
    Here, we use an index $a$ for an outcome of settings other than $x^*$ and $b$ for the setting $x^*$.
    Otherwise, we say $\mbM^{\cA|\cX}_\rA$ is star-incompatible with respect to $x^*$.
    (When it is clear from the context what $x^*$ is, we will simply say that the PMD $\mbM^{\cA|\cX}_\rA$ is star-compatible / star-incompatible.)
\end{definition}

\begin{fact}[Proposition 1 in Ref.~\cite{li2025necessary}]\label{fact}
The following holds: for an assemblage $\sigma^{\cA|\cX} _\rB$, $p^{x^*}_\guess(\sigma^{\cA|\cX} _\rB )<1$ if and only if $\sigma^{\cA|\cX} _\rB\notin \cD_\rB^{\cA|\cX}(x^*):=\{\sigma(\mbM^{\cA|\cX}_\rA,\rho_{\rA\rB})|\;\rho_{\rA\rB}\in\cS(\cH_\rA\otimes\cH_\rB),\mbM^{\cA|\cX}_\rA\text{is star-compatible}\}$.
\end{fact}
As summarized in Fact~\ref{fact}, an assemblage certifies randomness for the input $x^*$ if and only if its realization requires a star-incompatible PMD.
However, does \textit{every} star-incompatible PMD give rise to a random assemblage? 
It is conceivable that a star-incompatible PMD held by Alice is such that the assemblage $\sigma(\mbM^{\cA|\cX}_\rA,\rho_{\rA\rB})$, for any choice of $\rho_{\rA\rB}$, fails to be random. 
However, star-incompatibility of Alice's PMD is actually sufficient for certified randomness.

A key intermediary is \textit{star-unsteerability}, defined as follows.
\begin{definition}
    An assemblage $\sigma^{\cA|\cX} _\rB:=\{\sigma_\rB^{a|x}\}_{a\in\cA,x\in\cX}$ is said to be star-unsteerable (originally, $K_{1,n-1}(x^*)$-partially-unsteerable~\cite{li2025necessary}) with respect to setting $x^*$ if and only if there is a family of unnormalized states $\{\rho^{ab|x}_\rB\}_{a,b\in\cA,x\in\cX\setminus\{x^*\}}$ such that
    \begin{align}
        \sigma_\rB^{a|x}&=\sum_{b\in\cA}\rho^{ab|x}_\rB\;(
             \forall x\in\cX\setminus\{x^*\},\;\forall a\in\cA),\\
        \sigma^{b|x^*}_{\rB}&=\sum_{a\in\cA}\rho^{ab|x}_\rB \; (\forall x\in\cX\setminus\{x^*\},\;\forall b\in\cA).
    \end{align}
    Otherwise, the assemblage $\sigma^{\cA|\cX} _\rB $ is said to be star-steerable with respect to the ensemble $x^*$.
    (When it is clear from the context what $x^*$ is, we will simply say that the assemblage is star-steerable / star-unsteerable.)
\end{definition}
The following lemma, proven by Ref.~\cite{li2025necessary}, fully characterizes assemblages with no certified randomness.
\begin{lemma}\label{lemma:steerability_randomness}
    An assemblage $\sigma^{\cA|\cX}_{\rB }$ has no certified randomness, that is, $p^{x^*}_\guess(\sigma^{\cA|\cX}_{\rB })=1$, if and only if $\sigma^{\cA|\cX}_{\rB }$ is star-unsteerable.
\end{lemma}
We present the following lemma, which asserts that star-steerable assemblages can be generated from star-incompatible measurements.
A complete proof is provided in Supplemental Material~\cite{supplement}.
\begin{lemma}\label{lemma:compatibility-steerability}
    Let $\dim\cH_\rA=\dim\cH_\rB=d<+\infty$. If a PMD $\mbM^{\cA|\cX}_\rA$ is star-incompatible, then, for any pure entangled state $\ket{\psi}_{\rA\rB}$ with Schmidt rank $d$, $\sigma(\mbM^{\cA|\cX}_\rA,\ket{\psi}_{\rA\rB})$ is star-steerable.
\end{lemma}

Fact~\ref{fact} follows from combining Lemma~\ref{lemma:steerability_randomness} with the fact that generating star-steerable assemblages needs star-incompatible measurements~\cite{li2025necessary}. 
Here we combine Lemma~\ref{lemma:compatibility-steerability}—which states that star-incompatible measurements can produce star-steerable assemblages—with Lemma~\ref{lemma:steerability_randomness} to obtain the converse claim. This is summarized by the following proposition (also see Fig.~\ref{fig:equivalence}):
\begin{figure}[t]
\centering
\includegraphics[width=\linewidth]{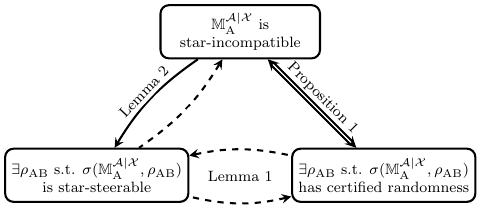}
\caption{Equivalence between star-incompatibility, star-steerability, and certified randomness. The dashed arrows indicate implications shown in Proposition 1 in Ref.~\cite{li2025necessary}, the solid arrows represent our results. The double left and right arrows indicate equivalence.}
\label{fig:equivalence}
\end{figure}

\begin{proposition}\label{proposition:equivalence}
    Let $\dim\cH_\rA=\dim\cH_\B=d$.
    For a given PMD $\mbM^{\cA|\cX}_\rA$ on $\cH_\rA$, the following statements are equivalent: (i) $\mbM^{\cA|\cX}_\rA$ is star-incompatible. (ii) There is a state $\rho_{\rA\rB}\in\cS(\cH_\rA\otimes\cH_\rB)$ such that $\sigma(\mbM^{\cA|\cX}_\rA,\rho_{\rA\rB})$ is star-steerable. (iii) There is a state $\rho_{\rA\rB}\in\cS(\cH_\rA\otimes\cH_\rB)$ such that $p^{x^*}_\guess(\sigma(\mbM^{\cA|\cX}_\rA,\rho_{\rA\rB}))<1$.
\end{proposition}

\begin{remark}
    It is known that pure entangled states with full Schmidt rank are useful for maximal randomness when measurements can be appropriately selected in both DI~\cite{acin2012randomness} and 1SDI~\cite{skrzypczyk2018maximal} setups. 
    In this paper, rather than states, we focus on characterizing sets of measurements that generate certified randomness in the 1SDI setting.
\end{remark}

\begin{remark}
    Confirming $p^{x^*}_\guess(\sigma(\mbM^{\cA|\cX}_\rA,\ket{\psi}_{\rA\rB}))=1$ for a full Schmidt rank pure entangled state $\ket{\psi}_{\rA\rB}$ with respect to all $x^*\in\cX$ is equivalent to certifying $\mbM^{\cA|\cX}_\rA$ is star-compatible for all $x^*\in\cX$, which in turn is equivalent to certifying that the joint measurability structure \cite{kunjwal2014quantum} of  $\mbM^{\cA|\cX}_\rA$ satisfies pairwise compatibility~\cite{liang2011specker,heunen2014quantum,kunjwal2014quantum}. This means that in any set of measurements that satisfies pairwise compatibility, each individual measurement is useless for steering-based randomness, even though such measurements (\textit{e.g.}, Specker's scenario \cite{liang2011specker}) can be useful for Bell inequality violations \cite{quintino2014joint,bene2018measurement,YAK24} as well as for demonstrations of generalized contextuality \cite{KG14,Kunjwal16}.
\end{remark}

\textit{Minimizing the guessing probability}---.
If Alice performs a star-incompatible PMD, the guessing probability calculated from the assemblage can be lower than $1$ by Proposition~\ref{proposition:equivalence}, but may vary depending on the state Alice and Bob share. 
We address the question of how one can make the guessing probability small for a fixed set of measurements by optimizing the shared state.
For this purpose, first, we introduce the following quantity.
\begin{definition}
    For a given PMD $\mbM^{\cA|\cX}_\rA$ and $\rho_\rA\in\cS(\cH_\rA)$, define a guessing probability $p^{x^*}_\guess (\mbM^{\cA|\cX}_\rA,\rho_\rA)$ as follows:
    \begin{align}
        p^{x^*}_\guess (\mbM^{\cA|\cX}_\rA,&\rho_\rA):=\maximize{G^{ae|x}_\rA}\quad \sum_{e\in\cA}\Tr[G^{ee|x^*}_\rA\rho_\rA] \notag\\
        \subto\quad&\sum_{e\in\cA}G^{ae|x}_\rA=M^{a|x}_\rA\;(\forall x\in\cX,\;\forall a\in\cA)\notag\\
        &\sum_{a\in\cA}G^{ae|x}_\rA=\sum_{b\in\cA}G^{be|x^*}_\rA\;(\forall x\in\cX,\;\forall e\in\cA)\notag\\
        &G^{ae|x}_\rA\succeq\zero\:(\forall x\in\cX,\;\forall a, e\in\cA).\label{eq:guessing_PMD}
    \end{align}
\end{definition}
The situation of Eq.~\eqref{eq:guessing_PMD} is as follows: the measurement device given by the PMD $\{G^{ae|x}_\rA\}_{a,e\in\cA,x\in\cX}$ leaks classical information $e$ to an attacker Eve after the measurement, and Eve guesses Alice's outcome $a$.
Unlike the 1SDI setup, Eve does not have a quantum system correlated with Alice's quantum system.
The first constraint means that insofar as Alice is concerned, the device performs as a PMD $\mbM^{\cA|\cX}_\rA$, \textit{i.e.,} Alice cannot detect that the attacker is eavesdropping on her outcomes.
The second constraint requires the attacker's POVM does not depend on the measurement setting. 
Therefore, if $\mbM^{\cA|\cX}_\rA$ is star-compatible, one can choose $G^{ae|x}_\rA$ so that $\sum_{e\in\cA}G^{ae|x}_\rA=M^{a|x}_\rA$ for all $x,a$ and $\sum_{a\in\cA}G^{ae|x}_\rA=M^{e|x^*}_\rA$, and specifically, one can set $G^{ae|x^*}_\rA=\delta_{ae}M^{a|x^*}_\rA$, which makes $p^{x^*}_\guess(\mbM^{\cA|\cX}_\rA,\rho_\rA)=1$\footnote{We mention that the guessing probability analyzed in Refs.~\cite{senno2023quantifying,curran2025maximal} considered an attacker having a correlated state with the auxiliary system of Alice's measurement. In our notation, the constraints they considered are $\sum_{a\in\cA}G^{ae|x}_\rA=p_e\one_\rA$, where $\{p_e\}_{e\in\cA}$ is a probability distribution, while in our case, it can be an arbitrary POVM.}.
The next proposition compares $p_\guess^{x^*}(\sigma(\mbM^{\cA|\cX}_\rA,\rho_{\rA\rB}))$ and $p^{x^*}_\guess (\mbM^{\cA|\cX}_\rA,\rho_\rA)$ (for the proof, see Supplemental Material~\cite{supplement}).

\begin{proposition}\label{proposition:assemblage_PMD}
    Let $\mbM^{\cA|\cX}_\rA$ be a PMD and $\rho_\mathrm{A}\in\cS(\cH_\rA)$ be a state of Alice.
    For any assemblage $\sigma(\mbM^{\cA|\cX}_\rA,\rho_{\rA\rB})$ such that $\Tr_\rB\rho_{\rA\rB}=\rho_\rA$, it holds that
    \begin{equation}
        p_\guess^{x^*}(\sigma(\mbM^{\cA|\cX}_\rA,\rho_{\rA\rB}))\ge p^{x^*}_\guess (\mbM^{\cA|\cX}_\rA,\rho_\rA).
    \end{equation}
\end{proposition}

For a given set of POVMs $\mbM^{\cA|\cX}_\rA$, computing $\inf_{\rho_{\rA\rB}} p_\guess^{x^*}(\sigma(\mbM^{\cA|\cX}_\rA,\rho_{\rA\rB}))$ is a nonlinear problem as it is, but we can obtain an SDP-computable lower bound by using $p^{x^*}_\guess (\mbM^{\cA|\cX}_\rA,\rho_\rA)$. Moreover, this lower bound is exact when Bob has a system whose Hilbert space dimension is no smaller than the dimension of Alice's system (for the proof, see Supplemental Material~\cite{supplement}).
\begin{theorem}\label{theorem:assemblage_PMD}
    Let $\mbM^{\cA|\cX}_\rA$ be a PMD.
    It holds that
    \begin{equation}
        \inf_{\rho_{\rA\rB}} p_\guess^{x^*}(\sigma(\mbM^{\cA|\cX}_\rA,\rho_{\rA\rB}))\ge\min_{\rho_\rA}p^{x^*}_\guess (\mbM^{\cA|\cX}_\rA,\rho_\rA),
    \end{equation}
    where the infimum and minimization are taken over all states; $\rho_{\rA\rB}\in\cS(\cH_\rA\otimes\cH_\rB)$ and $\rho_{\rA}\in\cS(\cH_\rA)$, respectively.
    Moreover, if $\dim\cH_\rA\le \dim\cH_\rB$, it holds that
    \begin{equation}
        \inf_{\rho_{\rA\rB}} p_\guess^{x^*}(\sigma(\mbM^{\cA|\cX}_\rA,\rho_{\rA\rB}))=\min_{\rho_\rA}p^{x^*}_\guess (\mbM^{\cA|\cX}_\rA,\rho_\rA).
    \end{equation}
\end{theorem}

\textit{Star-incompatibility estimation from randomness}---.
Several studies aim to detect or quantify the incompatibility in a device-independent manner or in a one-sided device-independent manner~\cite{wolf2009measurements,chen2016natural,cavalcanti2016quantitative,quintino2019device,chen2021device}. 
Here we provide a method to quantify the amount of star-incompatibility necessary for getting an assemblage with a given guessing probability $p_\guess^{x^*}(\sigma(\mbM^{\cA|\cX}_\rA,\rho_{\rA\rB}))$.
We introduce \textit{star-incompatibility weight}, based on incompatibility weight~\cite{pusey2015verifying}, to quantify the star-incompatibility.

\begin{definition}\label{definition:weight}
    For a given PMD $\mbM^{\cA|\cX}_\rA$ and setting $x^*\in 
    \cX$, star-incompatibility weight $W^{x^*}(\mbM^{\cA|\cX}_\rA)$ is defined as follows:
    \begin{align}
        W^{x^*}(\mbM^{\cA|\cX}_\rA)&:=\minimize{F^{a|x}_\rA,E^{a|x}_\rA}\quad  w\ge 0\notag\\
        \subto\quad & M^{a|x}_\rA=(1-w)F^{a|x}_\rA+wE^{a|x}_\rA\notag\\
        &\hspace{25mm}(\forall x\in\cX,\forall a\in\cA)\notag\\
        \quad & \F^{\cA|x}_\rA,\E^{\cA|x}_\rA\in\sM^{\cA}_\rA\quad(\forall x\in\cX)\notag\\
        \quad & \{\F^{\cA|x}_\rA\}_{x\in\cX}\text{ is star-compatible}.\label{eq:weight}
    \end{align}
\end{definition}

Weight-based measures have also been introduced for steering~\cite{skrzypczyk2014quantifying} and star-steering~\cite{li2025necessary}, and are valid in general convex quantum resource theories, with a clear operational meaning~\cite{ducuara2020operational}. We have the following theorem that quantifies the necessary amount of star-incompatibility weight for a given assemblage (for the proof, see Supplemental Material~\cite{supplement}).
\begin{theorem}\label{theorem:weight}
    Let $\mbM^{\cA|\cX}_\rA$ be a PMD.
    For any state $\rho_{\rA\rB}\in\cS(\cH_\rA\otimes\cH_\rB)$, it holds that
    \begin{equation}\label{eq:lower_bound}
        W^{x^*}(\mbM^{\cA|\cX}_\rA)\ge \frac{|\cA|}{|\cA|-1}\{1-p_\guess^{x^*}(\sigma(\mbM^{\cA|\cX}_\rA,\rho_{\rA\rB}))\}.
    \end{equation}
\end{theorem}
The inequality of Eq.~\eqref{eq:lower_bound} is tight, as shown in Example~\ref{example:1-w/2} in the End Matter.
However, a large star-incompatibility weight does not always imply a small guessing probability (see Example~\ref{example:unbiased} in the End Matter).

By definition, $W^{x^*}(\mbM^{\cA|\cX}_\rA)=0$ if and only if $\mbM^{\cA|\cX}_\rA$ is star-compatible, and its maximum value is $1$.
The presence of certified randomness, \textit{i.e.}, $p^{x^*}_\guess(\sigma^{\cA|\cX}_\rB)<1$, makes the right-hand side of this inequality positive, which forces $W^{x^*}(\mbM^{\cA|\cX}_\rA)$ to be positive, detecting star-incompatibility.
Thus, certified randomness is not only a witness of star-incompatibility \cite{li2025necessary} but it also provides a quantitative one-sided estimation of star-incompatibility.
For example, maximum randomness as considered in Ref.~\cite{skrzypczyk2018maximal}, \textit{i.e.}, $p^{x^*}_\guess(\sigma(\mbM^{\cA|\cX}_\rA,\rho_{\rA\rB}))=1/|\cA|$, makes the right-hand side of this inequality $1$.
Thus, maximum randomness is achieved only for maximally star-incompatible PMDs in the sense that star-incompatible weight is $1$, the maximum possible value.

\textit{Conclusions.}---In the 1SDI randomness certification setup~\cite{passaro2015optimal}, while previous work~\cite{li2025necessary} implies that star-incompatibility is a necessary resource to generate assemblages with certified randomness, we have shown that all star-incompatible PMDs can generate assemblages with certified randomness by properly choosing a state.
This establishes the equivalence between certified randomness in the 1SDI setup and star-incompatibility.
We further analyzed the minimum guessing probability for a fixed measurement set by optimizing the shared state.
We provide an SDP-computable lower bound that becomes exact when Bob's Hilbert-space dimension is not smaller than Alice's.
We also give a lower bound on the star-incompatibility weight using the guessing probability.
This result demonstrates that a one-sided estimation of the amount of star-incompatibility using guessing probability is possible, providing a resource-theoretic link between star-incompatibility and  randomness.

It remains an open question whether the 1SDI results on randomness vs.~incompatibility here can be extended to the fully DI case. However, given that incompatibility does not always imply Bell nonlocality~\cite{quintino2014joint,bene2018measurement,hirsch2018quantum}, and that there is no simple quantitative relationship between randomness and Bell nonlocality either~\cite{acin2012randomness}, the results of this study might not straightforwardly extend to the DI setup in general.

Finally, we note that the property of star-compatibility has a curious connection with the conjecture that all joint measurability structures are realizable with qubit measurements (Conjecture $4$ in Ref.~\cite{AK20}): specifically, if one can construct a PMD with $5$ or more qubit measurements that is not only star-compatible but also fails to have any further compatibility relations, a potential counterexample to this conjecture would be removed; alternatively, the impossibility of such a PMD would rule out the conjecture. Whether some insights from quantum steering, via its connection with star-compatibility, could address this challenge is an interesting open question that we leave for future work.

\textit{Acknowledgments}---.
The authors thank Edwin Lobo and Shashaank Khanna for insightful discussions. This work received support from the French government under the France 2030 investment plan, as part of the Initiative d'Excellence d'Aix-Marseille Université-A*MIDEX, AMX-22-CEI-01.

OpenAI ChatGPT (GPT-5.5, GPT-5.6 Sol, and GPT-6 Astra) assisted in improving the text and developing the proof strategy of Theorem~\ref{theorem:assemblage_PMD}, the Examples~\ref{example:1-w/2} and~\ref{example:unbiased}, and in generating a code~\cite{code} for Fig.~\ref{fig:example}.
The authors take full responsibility for the contents of the paper.

\bibliography{myref}

\renewcommand{\theequation}{E\arabic{equation}}
\setcounter{equation}{0}  

\onecolumngrid
\section*{End Matter}
\twocolumngrid

\textit{Appendix A: An SDP for $p^{x^*}_\guess (\mbM^{\cA|\cX}_\rA,\rho_\rA)$}---.
Computing $\min_{\rho_\rA\in\cS(\cH_\rA)}p^{x^*}_\guess (\mbM^{\cA|\cX}_\rA,\rho_\rA)$ is not an SDP as stated, but the dual problem of $p^{x^*}_\guess (\mbM^{\cA|\cX}_\rA,\rho_\rA)$ (provided in Supplemental Material~\cite{supplement}) turns $\min_{\rho_\rA\in\cS(\cH_\rA)}p^{x^*}_\guess (\mbM^{\cA|\cX}_\rA,\rho_\rA)$ into the following SDP:
\begin{align}
        \minimize{X^{a|x}_\rA,Y^{e|x}_\rA,\rho_\rA}\quad&\sum_{x\in\cX}\sum_{a\in\cA}\Tr[X^{a|x}_\rA  M^{a|x}_\rA]\notag\\
        \subto \quad&X^{a|x}_\rA -Y^{e|x}_\rA+\delta_{xx^*}\sum_{x'\in\cX}Y^{e|x'}_\rA\notag\\
        \quad&\succeq\delta_{xx^*}\delta_{ae}\rho_\rA \quad(\forall x\in\cX,\;\forall a,\;e\in\cA),\notag\\
        & \rho_\rA\succeq\zero_\rA,\;\Tr\rho_\rA=1.\label{eq:min_state_SDP}
\end{align}

\textit{Appendix B: tightness of Eq.~\eqref{eq:lower_bound}}---.
In this section, $\rA$ is a qubit system and $\rB$ is a system of (finite) dimension at least $2$.
We provide an example that achieves the equality in Eq.~\eqref{eq:lower_bound}.
\begin{example}\label{example:1-w/2}
Consider the following set of POVMs $\mbM^{\cA|\cX}_\rA(w)=\{M^{a|x}_\rA\}_{a\in\cA,x\in\cX}$ on a qubit system $\rA$:
\begin{equation}
\begin{split}
    M^{0|0}_\rA(w)&:=(1-w)\one_\rA+w\ketbra{0}{0}_\rA,\\
    M^{1|0}_\rA(w)&:=w\ketbra{1}{1}_\rA\quad(0\le w\le 1),\\
    M^{0|1}_\rA(w)&:=\ketbra{+}{+}_\rA,\;M^{1|1}_\rA(w):=\ketbra{-}{-}_\rA,
\end{split}
\end{equation}
where $\{\ket{0}_\rA,\ket{1}_\rA\}$ is a standard basis of $\cH_\rA$, and $\ket{+}_\rA:=\frac{1}{\sqrt{2}}(\ket{0}_\rA+\ket{1}_\rA)$ and $\ket{-}_\rA:=\frac{1}{\sqrt{2}}(\ket{0}_\rA-\ket{1}_\rA)$.
The dimension of Bob's Hilbert space is $2$ or higher.
As shown in Supplemental Material~\cite{supplement}, we have $W^0(\mbM^{\cA|\cX}_\rA(w))=w$ and $\min_{\rho_{\rA\rB}}p^0_\guess(\sigma(\mbM^{\cA|\cX}_\rA(w),\rho_{\rA\rB}))=1-\frac{w}{2}$, which attains the equality of Eq.~\eqref{eq:lower_bound} for all $0\le w\le 1$.
\end{example}

\begin{example}\label{example:unbiased}
    Consider the following set of POVMs $\mbN^{\cA|\cX}_\rA(w):=\{N^{a|x}_\rA\}_{a\in\cA,x\in\cX}$ on a qubit system $\rA$:
    \begin{equation}
    \begin{split}
        N^{0|0}_\rA(w)&=\ketbra{0}{0}_\rA,\quad N^{1|0}_\rA(w)=\ketbra{1}{1}_\rA,\\
        N^{0|1}_\rA(w)&=\frac{1}{2}\left(\one_\rA+wX_\rA\right),\\
        N^{1|1}_\rA(w)&=\frac{1}{2}\left(\one_\rA-wX_\rA\right)\quad(0\le w\le 1),
    \end{split}
    \end{equation}
    where $\{\ket{0}_\rA,\ket{1}_\rA\}$ is the standard basis of $\cH_\rA$.
    The dimension of Bob's Hilbert space is $2$ or higher.
    As shown in Supplemental Material~\cite{supplement}, it holds that $W^0(\mbN^{\cA|\cX}_\rA(w))=w$ and $\min_{\rho_{\rA\rB}}p^0_\guess(\sigma(\mbN^{\cA|\cX}_\rA(w),\rho_{\rA\rB}))=\frac{1+\sqrt{1-w^2}}{2}$.
\end{example}

Let us compare Examples~\ref{example:1-w/2} and~\ref{example:unbiased}.
For example, when the star-incompatibility weight is $0.6$, $\min_{\rho_{\rA\rB}}p^0_\guess(\sigma(\mbM^{\cA|\cX}_\rA(w),\rho_{\rA\rB}))=0.7$, and when the star-incompatibility weight is $0.8$, $\min_{\rho_{\rA\rB}}p^0_\guess(\sigma(\mbN^{\cA|\cX}_\rA(w),\rho_{\rA\rB}))=0.8$.
Therefore, just because the weight is larger, it does not necessarily mean that the guessing probability is smaller.

\textit{Appendix C: See-saw algorithm}---.
Here, we introduce a method called the see-saw algorithm, which searches for a state to make an assemblage as random as possible for a given PMD. 
Our algorithm is similar to one in Ref.~\cite{passaro2015optimal}, except we fix a PMD and move a state.
We use the dual problem of Eq.~\eqref{eq:guessing_assemblage} given as follows~\cite{passaro2015optimal}:
\begin{align}
        \minimize{X^{a|x}_\rB,Y^{e|x}_\rB}\quad
        &\sum_{x\in\cX}\sum_{a\in\cA}\Tr[X^{a|x}_\rB\sigma^{a|x}_\rB]\notag\\
    \subto\quad & X^{a|x}_\rB-Y^{e|x}_\rB+\delta_{xx^*}\sum_{x'\in\cX}Y^{e|x'}\notag\\
    &\succeq\delta_{xx^*}\delta_{ae}\one_\rB\;(\forall x\in\cX,\forall a,e\in\cA).\label{eq:dual_guessing_assemblage}
\end{align}

The algorithm goes as follows (Algorithm~\ref{algorithm:see-saw}).
First, we randomly chose a state $\rho_{\rA\rB}\in\cS(\cH_\rA\otimes\cH_\rB)$ and construct an assemblage $\sigma^{\cA|\cX}_\rB:=\sigma(\mbM^{\cA|\cX}_\rA,\rho_{\rA\rB})$.
Then we solve Eq.~\eqref{eq:dual_guessing_assemblage}, the dual problem of Eq~.\eqref{eq:guessing_assemblage}.
Next, we compute $\underset{\rho_{\rA\rB}\in\cS(\cH_\rA\otimes\cH_\rB)}{\mathrm{arg}\min}\sum_{x\in\cX}\sum_{a\in\cA}\Tr[(M^{a|x}_\rA\otimes X^{\circ a|x}_\rB)\rho_{\rA\rB}]$.
For the state $\rho^\circ_{\rA\rB}$ explored in this optimization, we create an assemblage $\sigma(\mbM^{\cA|\cX}_\rA,\rho^\circ_{\rA\rB})$, and solve Eq.~\eqref{eq:dual_guessing_assemblage} again.
This iteration is repeated until the guessing probability change and state change become sufficiently small in terms of the distance based on the Hilbert--Schmidt norm $\|\cdot\|_{\rH\rS}$ (see, e.g., \cite{wilde2017quantum}).

\begin{figure}[t]
    \centering
    \includegraphics[width=\linewidth]{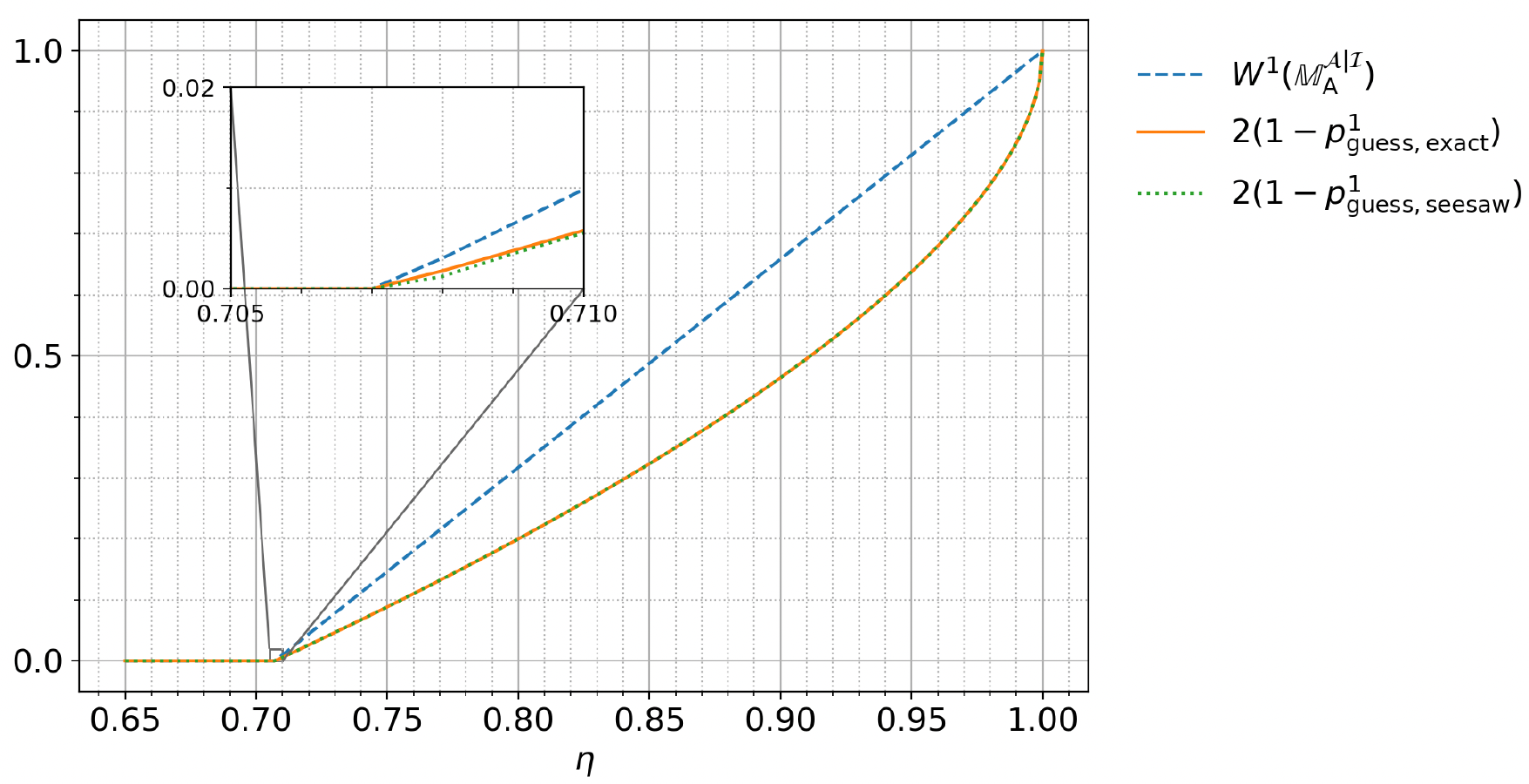}
    \caption{Take a PMD $\mbM^{\cA|\cI}_{\rA}:=\{\M^{\cA|i}_\rA\}_{i=1}^3$ where $\M^{\cA|i}_\rA$ is given by Eq.~\eqref{eq:pauli}. The dashed line is the star-incompatibility weight $W^1(\mbM^{\cA|\cI}_\rA)$.
    The solid curve is $2(1-\min_{\rho_{\rA\rB}}p^1_\guess(\sigma(\mbM^{\cA|\cI}_\rA,\rho_{\rA\rB}))$ computed by an SDP Eq.~\eqref{eq:min_state_SDP}.
    The dotted curve is $2(1-p^1_\guess(\sigma(\mbM^{\cA|\cI}_\rA,\rho^{\mathrm{see-saw}}_{\rA\rB}))$ where $\rho^{\mathrm{see-saw}}_{\rA\rB}$ is obtained by Algorithm~\ref{algorithm:see-saw}.}
    \label{fig:example}
\end{figure}

\begin{example}
    Let $\dim\cH_\rA=2$ and $\cH_\rB$ be a finite-dimensional Hilbert space of dimension $2$ or higher.
    We consider binary parametrized POVMs on $\cH_\rA$ that consist of Pauli matrices $\sigma_0:=I,\sigma_1:=X,\sigma_2:=Y,\sigma_3:=Z$
    \begin{equation}\label{eq:pauli}
        \M^{\cA|i}_\rA:=\left\{\frac{1}{2}\sigma_0\pm\frac{1}{2}\eta\sigma_i\right\}\;(i\in\{1,\dots,3\}).
    \end{equation}
    Here the parameter $\eta$ is $0\le\eta\le 1$.
\end{example}
The PMD we consider is $\mbM^{\cA|\cI}_\rA:=\{\M^{\cA|i}_\rA\}_{i=1}^3$.
As shown in Refs.~\cite{heinosaari2008notes,liang2011specker}, $\mbM^{\cA|\cI}_\rA:=\{\M^{\cA|i}_\rA\}_{i=1}^3$ is star-compatible with respect to $x^*=1$ when $\eta\le1/\sqrt{2}\approx0.7071$.
Fig.~\ref{fig:example} plots the star-incompatible weight and the lower bounds based on Eq.~\eqref{eq:lower_bound}; one uses the SDP Eq.~\eqref{eq:min_state_SDP} for minimization over states, which is exact because of Theorem~\ref{theorem:assemblage_PMD}.
The other one uses Algorithm~\ref{algorithm:see-saw} with $\dim\cH_\rB=5$, $\epsilon_p=1.0\times10^{-8}, \epsilon_\rho:=1.0\times10^{-10}$, $T_{\max}=4$, and $K=2$.
For each $\eta$, we run Algorithm~\ref{algorithm:see-saw} twice.
When $\eta=0.65$, two see-saw runs were performed, each initialized with an independently chosen random pure state.
For every subsequent $\eta$ ($\eta$ is in increments of $0.001$), one run was initialized with the best state obtained at the preceding value of $\eta$, while the other was initialized with a randomly chosen pure state.
The better of the two resulting guessing probabilities was retained.
All SDPs were formulated using CVXPY~\cite{cvxpy} and solved with MOSEK~\cite{mosek}.
The code is available at~\cite{code}.
The version is archived under the tag
``v0.2-preprint.''

\begin{widetext}
\begin{algorithm}
\caption{See-saw minimization}
\label{algorithm:see-saw}
\SetKwInOut{Input}{Input}
\SetKwInOut{Output}{Output}
\SetKwInOut{Parameter}{Parameter}

\Input{
A PMD
$\mbM^{\cA|\cX}_\rA:=\{\M^{\cA|x}_\rA\}_{x\in\cX}$,
$x^*\in\cX$, and an initial state
$\rho_{\rA\rB}^{(0)}$.
}
\Parameter{
$\epsilon_p,\epsilon_\rho>0$, and $T_{\max},K$ (positive integers).
}
\Output{
A state $\rho_{\rA\rB}^\star$ and
$p^{x^*}_\guess(
\sigma(\mbM^{\cA|\cX}_\rA,\rho^\star_{\rA\rB})
)$.
}

\BlankLine
\SetAlgoLined

Set
$t\leftarrow 0$,
$k\leftarrow 0$,
$p^\prev\leftarrow +\infty$, and
$\rho^\prev_{\rA\rB}\leftarrow\rho_{\rA\rB}^{(0)}$.\\

\While{$t<T_{\max}$}{

  Solve Eq.~\eqref{eq:dual_guessing_assemblage}
  for the assemblage
  $\sigma(\mbM^{\cA|\cX}_\rA,\rho^\prev_{\rA\rB})$ to obtain $p^{x^*}_\guess(\sigma(\mbM^{\cA|\cX}_\rA,\rho^\prev_{\rA\rB})
  )$ and a dual optimal solution $\{X^{\circ a|x}_\rB\}$.\\

  Set $p^\new\leftarrow p^{x^*}_\guess(\sigma(\mbM^{\cA|\cX}_\rA,\rho^\prev_{\rA\rB}))$.

  Set $\rho^\new_{\rA\rB}\leftarrow\underset{\rho_{\rA\rB}}{\mathrm{arg}\min}\sum_{x,a}\Tr[(M^{a|x}_\rA\otimes X^{\circ a|x}_\rB)\rho_{\rA\rB}].$

  \eIf{
    $|p^\new-p^\prev|<\epsilon_p$
    \textbf{and}
    $\|\rho_{\rA\rB}^\new-\rho^\prev_{\rA\rB}\|_{\rH\rS}<\epsilon_\rho$
  }{
    Set $k\leftarrow k+1$.\\
  }{
    Set $k\leftarrow 0$.\\
  }

  Set $p^\prev\leftarrow p^\new$ and $\rho^\prev_{\rA\rB}\leftarrow\rho_{\rA\rB}^\new$.\\

  Set $t\leftarrow t+1$.\\

  \If{$k\ge K$}{
    \textbf{break}.
  }
}

Set
$\rho^\star_{\rA\rB}\leftarrow\rho^\prev_{\rA\rB}$.\\

Solve Eq.~\eqref{eq:dual_guessing_assemblage} and set $p^\star
\leftarrow
p^{x^*}_\guess(
\sigma(
\mbM^{\cA|\cX}_\rA,
\rho^\star_{\rA\rB}
)
)$.

\Return{
$\rho^\star_{\rA\rB},\;p^\star$
}
\end{algorithm}
\end{widetext}

\clearpage
\newpage

\setcounter{page}{1}
\setcounter{section}{0}
\renewcommand{\thesection}{\arabic{section}}
\setcounter{equation}{0}
\renewcommand{\theequation}{S\arabic{equation}}
\setcounter{lemma}{0}
\renewcommand{\thelemma}{\arabic{lemma}}

\onecolumngrid

\begin{center}
    {\large \textbf{Supplemental Material for\\ ``All star-incompatible measurements can certify steering-based randomness''}}\\
    \vspace{0.3cm}
    Shintaro Minagawa\\
    \vspace{0.1cm}
    {\small\textit{Aix-Marseille University, CNRS, LIS, Marseille, France}}\\
    \vspace{0.3cm}
    Ravi Kunjwal\\
    \vspace{0.1cm}
    {\small
    \textit{Universit\'e libre de Bruxelles, QuIC, Brussels, Belgium and\\
    Aix-Marseille University, CNRS, LIS, Marseille, France}}
\end{center}

\section{Equivalence between star-incompatibility and certified randomness in the 1SDI setup}

This section provides the proof for Lemma~\ref{lemma:compatibility-steerability}.
The proof is a straightforward application of the method of Ref.~\cite{quintino2014joint}.

\begin{proof}[Proof of Lemma~\ref{lemma:compatibility-steerability}]
    In general, a state $\ket{\psi}_{\rA\rB}\in \cH_\rA\otimes\cH_\rB$, with Schmidt rank $d$, can be written as $\sum_{i=0}^{d-1}\sqrt{q _i}\ket{i}_\rA\otimes\ket{i}_\rB$ where $\{\ket{i}_\rA\}_{i=0}^{d-1}$ and $\{\ket{i}_\rB\}_{i=0}^{d-1}$ are basis of $\cH_\rA$ and $\cH_\rB$, respectively,  $q _i>0$ for all $i\in\{0,\dots,d-1\}$ and $\sum_{i=0}^{d-1}q _i=1$.
    Let $\ket{\Phi}_{\rA\rB}:=\sum_{i=0}^{d-1}\frac{1}{\sqrt{d}}\ket{i}_\rA\otimes\ket{i}_\rB$ be the maximally entangled state.
    Define an operator $D_\rA:=\sqrt{d}\;\mathrm{diag}(\sqrt{q _0},\dots,\sqrt{q_{d-1}})\in\cL(\cH_\rA)$.
    \begin{equation}
    \begin{split}    (D_\rA\otimes\one_\rB)\ket{\Phi}_{\rA\rB}&=\left(\sum_{i=0}^{d-1}\sqrt{q _i}\ketbra{i}{i}_\rA\otimes\one_\rB\right)\sum_{j=0}^{d-1}\ket{j}_\rA\otimes\ket{j}_\rB\\
    &=\sum_{i=0}^{d-1}\sqrt{q _i}\ket{i}_\rA\otimes\ket{i}_\rB\\
    &=\ket{\psi}_{\rA\rB}.
    \end{split}
    \end{equation}
    By choosing the state $\rho_{\rA\rB}$ that Alice and Bob share as $\ketbra{\psi}{\psi}_{\rA\rB}$, the assemblage $\{\sigma_\rB^{a|x}\}_{a\in\cA,x\in\cX}$ is given by
    \begin{equation}
        \begin{split}
            \sigma^{a|x}_\rB&=\Tr_\rA\left[\left(M^{a|x}_\rA\otimes\one_\rB\right)\ketbra{\psi}{\psi}_{\rA\rB}\right]\\
            &=\Tr_\rA\left[\left(M^{a|x}_\rA\otimes\one_\rB\right)(D_\rA\otimes \one_\rB)\ketbra{\Phi}{\Phi}_{\rA\rB}(D_\rA\otimes \one_\rB)\right]\\
            &=\Tr_\rA\left[(D_\rA\otimes \one_\rB)\left(M^{a|x}_\rA\otimes\one_\rB\right)(D_\rA\otimes \one_\rB)\ketbra{\Phi}{\Phi}_{\rA\rB}\right]\\
            &=\Tr_\rA\left[\left(D_\rA M^{a|x}_\rA D_\rA\otimes\one_\rB\right)\ketbra{\Phi}{\Phi}_{\rA\rB}\right]\\
            &=\Tr_\rA\left[\left(\one_\rA\otimes\left(D_\rA M^{a|x}_\rA D_\rA\right)^\top\right)\ketbra{\Phi}{\Phi}_{\rA\rB}\right]\\
            &=\Tr_\rA\left[\left(\one_\rA\otimes \left(D_\rA M^{a|x\top}_\rA D_\rA\right)\right)\ketbra{\Phi}{\Phi}_{\rA\rB}\right]\\
            &=\frac{1}{d}D_\rA M^{a|x\top}_\rA D_\rA.
        \end{split}
    \end{equation}
    The transpose is with respect to the basis of the Schmidt decomposition.
    The third line uses the cyclicity property. In general, the cyclicity property of the trace cannot be applied to the partial trace. But it holds that \begin{equation}\label{eq:ptrace_cyclicity}
        \Tr_\rA[X_{\rA\rB}(Y_\rA\otimes\one_\rB)(Z_\rA\otimes\one_\rB)]=\Tr_\rA[(Z_\rA\otimes\one_\rB)X_{\rA\rB}(Y_\rA\otimes\one_\rB)]=\Tr_\rA[(Y_\rA\otimes\one_\rB)(Z_\rA\otimes\one_\rB)X_{\rA\rB}].
    \end{equation} See, e.g., \cite[Exercise 4.3.15]{wilde2017quantum}.
    The fifth line uses the well-known property of maximally entangled state: $(X_\rA\otimes\one_\rB)\ket{\Phi}_{\rA\rB}=(\one_\rA\otimes X_\rB^\top)\ket{\Phi}_{\rA\rB}$.
    
    Since $D_\rA$ is full-rank and diagonal, $D^{-1}_\rA$ exists and $D^{-1\top}_\rA=D^{-1}_\rA$.
    Thus, there is a relation
    \begin{equation}
        M^{a|x}_\rA=dD^{-1}_\rA\sigma^{a|x\top}_\rB D^{-1}_\rA\quad(\forall x\in\cX,\forall a\in\cA).
    \end{equation}
    Since it holds that $\sum_{a\in\cA}M^{a|x}_\rA=\one$ because $\{M^{a|x}_\rA\}_{a\in\cA}$ is a POVM, it holds that
    \begin{equation}\label{eq:assemblage_normalization}
        \begin{split}
            \sum_{a\in\cA}dD^{-1}_\rA\sigma^{a|x\top}_\rB D^{-1}&=dD^{-1}_\rA\left(\sum_{a\in\cA}\sigma^{a|x\top}_\rB\right)D^{-1}_\rA\\
            &=dD^{-1}_\rA\rho_B^\top D^{-1}_\rA\\
            &=\one_\rA
        \end{split}
    \end{equation}
    where $\rho_\rB:=\Tr_\rA\ketbra{\psi}{\psi}_{\rA\rB}=\sum_{a\in\cA}\sigma^{a|x}_\rB\;(\forall x\in\cX)$.
    
    Now, assume that $\{\sigma_\rB^{a|x}\}_{a\in\cA,x\in\cX}$ is star-unsteerable.
    Then, there are states $\{\rho^{ab|x}_\rB\}_{a,b\in\cA ,x\in\cX\setminus\{x^*\}}$ such that
    \begin{align}
        \sigma_\rB^{a|x}&=\sum_{b\in\cA}\rho^{ab|x}_\rB\quad(\forall x\in\cX\setminus\{x^*\},\forall a\in\cA),\\
        \sigma^{b|x^*}_\rB&=\sum_{a\in\cA}\rho^{ab|x}_\rB\quad(\forall x\in\cX\setminus\{x^*\},\forall b\in\cA).
    \end{align}
    where
    \begin{equation}\label{eq:assemblage_condition}
        \sum_{a,b\in\cA}\rho^{ab|x}_\rB=\rho_\rB\quad(\forall x\in\cX\setminus\{x^*\}).
    \end{equation}
    
    Define the following operator
    \begin{equation}
        G^{ab|x}_\rA:=d D^{-1}_\rA\rho^{ab|x\top}_\rB  D^{-1}_\rA\quad(\forall x\in\cX\setminus\{x^*\},\forall a,b\in\cA).
    \end{equation}
    First, $G^{ab|x}_\rA\succeq\zero_\rA\;(\forall x\in\cX\setminus\{x^*\},\forall a,b\in\cA)$ because $\rho^{ab|x\top}_\rB\succeq\zero_\rB$.
    Also, we have
    \begin{equation}
        \begin{split}
            \sum_{a,b\in\cA }G^{ab|x}_\rA&=dD^{-1}_\rA\left(\sum_{a,b\in\cA }\rho^{ab|x\top}_\rB\right)D^{-1}_\rA\\
            &=dD^{-1}_\rA\rho^\top_\rB D^{-1}_\rA\\
            &=\one_\rA\quad(\forall x\in\cX).
        \end{split}
    \end{equation}
    The second line uses \eq{eq:assemblage_condition} and the third line uses \eqref{eq:assemblage_normalization}.
    Therefore, $\{G^{ab|x}_\rA\}_{a,b\in\cA }$ is a POVM.
    Also, $\{G^{ab|x}_\rA\}_{a,b\in\cA }$ satisfies
    \begin{align}
         \sum_{b\in\cA}G^{ab|x}_\rA&=\sum_{b\in\cA}dD^{-1}_\rA\rho^{ab|x\top}_\rB D^{-1}_\rA=dD^{-1}_\rA\sigma^{a|x\top}_\rB D^{-1}_\rA=M^{a|x}_\rA\quad(\forall x\in\cX\setminus\{x^*\},\forall a\in\cA),\\
        \sum_{a\in\cA}G^{ab|x}_\rA&=\sum_{a\in\cA}dD^{-1}_\rA\rho^{ab|x\top}_\rB D^{-1}_\rA=dD^{-1}_\rA\sigma^{b|x^*\top}_\rB D^{-1}_\rA=M^{b|x^*}_\rA\quad(\forall b\in\cA),
    \end{align}
    which means that $\mbM^{\cA|\cX}_\rA$ is star-compatible.
\end{proof}

\section{Universal lower bound on the guessing probability}
\subsection{The dual problem of $p^{x^*}_\guess(\mbM^{\cA|\cX}_\rA,\rho_\rA)$}
Define Lagrangian as follows:
\begin{equation}
    \begin{split}
        &L(X^{a|x}_\rA ,Y^{e|x}_\rA ,Z^{ae|x}_\rA;G^{ae|x}_\rA)\\
        &:=
        \begin{dcases}
            &\sum_{x\in\cX}\sum_{a,e\in\cA}\Tr[G^{ae|x}_\rA(\delta_{xx^*}\delta_{ae}\rho_\rA)]+\sum_{x\in\cX}\sum_{a\in\cA}\Tr\left[X^{a|x}_\rA \left(M^{a|x}_\rA -\sum_{e\in\cA}G^{ae|x}_\rA \right)\right]\\
            &+\sum_{x\in\cX}\sum_{e\in\cA}\Tr\left[Y^{e|x}_\rA \left(\sum_{a\in\cA}G^{ae|x}_\rA -\sum_{b\in\cA}G^{be|x^*}_\rA\right)\right]+\sum_{x\in\cX}\sum_{a,e\in\cA}\Tr[Z^{ae|x}_\rA G^{ae|x}_\rA ]\quad(Z^{ae|x}_\rA\succeq\zero_\rA)\;,\\
            &+\infty\quad(\text{otherwise})
        \end{dcases}
    \end{split}
\end{equation}
where $X^{a|x}_\rA ,Y^{e|x}_\rA ,Z^{ae|x}_\rA\in\cL_\rH(\cH_\rA)$ are Lagrange multipliers, which are Hermitian operators.
Note that $G^{ae|x}_\rA$ is also just Hermitian here, while Eq.~\eqref{eq:guessing_PMD} has a constraint $G^{ae|x}_\rA\succeq\zero_\rA$. 
This is because our aim is to transform an optimization problem with constraints into an unconstrained one by incorporating the constraints into the definition of the Lagrangian.
Indeed, thanks to the condition that the Lagrangian is $+\infty$ when $Z^{ae|x}_\rA\nsucceq\zero_\rA$, when taking the infimum of the Lagrangian with respect to the Lagrange multipliers, $Z^{ae|x}_\rA\succeq\zero_\rA$ is imposed.
Under these circumstances, if $G^{ae|x}_\rA\nsucceq\zero_\rA$, taking the infimum of the Lagrangian with respect to $Z^{ae|x}_\rA$ results in $-\infty$.
Therefore, we get
\begin{equation}
    \inf_{X^{a|x}_\rA ,Y^{e|x}_\rA ,Z^{ae|x}_\rA}L(X^{a|x}_\rA ,Y^{e|x}_\rA ,Z^{ae|x}_\rA,G^{ae|x}_\rA)=
    \begin{dcases}
        \sum_{x\in\cX}\sum_{a,e\in\cA}\Tr[G^{ae|x}_\rA(\delta_{xx^*}\delta_{ae}\rho_\rA)]\quad&(G^{ae|x}_\rA\succeq\zero_\rA)\;,\\
        -\infty\quad&(\text{otherwise})\;,
    \end{dcases}
\end{equation}
which allows us to rewrite the primal problem Eq.~\eqref{eq:guessing_PMD} as follows:
\begin{equation}
    \sup_{G^{ae|x}_\rA}\inf_{X^{a|x}_\rA ,Y^{e|x}_\rA ,Z^{ae|x}_\rA}L(X^{a|x}_\rA ,Y^{e|x}_\rA ,Z^{ae|x}_\rA,G^{ae|x}_\rA).
\end{equation}

We obtain the dual problem by changing the order of inf and sup.
\begin{equation}
    \inf_{X^{a|x}_\rA ,Y^{e|x}_\rA ,Z^{ae|x}_\rA}\sup_{G^{ae|x}_\rA}L(X^{a|x}_\rA ,Y^{e|x}_\rA ,Z^{ae|x}_\rA,G^{ae|x}_\rA).
\end{equation}
To evaluate $\sup_{G^{ae|x}_\rA}L(X^{a|x}_\rA ,Y^{e|x}_\rA ,Z^{ae|x}_\rA,G^{ae|x}_\rA)$, let us rewrite the Lagrangian when $Z^{ae|x}_\rA\succeq\zero_\rA$ as follows:
\begin{equation}
    \begin{split}
        L(X^{a|x}_\rA ,Y^{e|x}_\rA ,Z^{ae|x}_\rA,G^{ae|x}_\rA )&=\sum_{x\in\cX}\sum_{a\in\cA}\Tr[X^{a|x}_\rA M^{a|x}_\rA ]\\
        &\quad+\sum_{x\in\cX}\sum_{a,e\in\cA}\Tr\left[\left(\delta_{xx^*}\delta_{ae}\rho_\rA -X^{a|x}_\rA +Y^{e|x}_\rA -\delta_{xx^*}\sum_{x'\in\cX}Y^{e|x'}_\rA +Z^{ae|x}_\rA \right)G^{ae|x}_\rA \right].
    \end{split}
\end{equation}
Here, again, $G^{ae|x}_\rA$ is Hermitian.
Thus, if 
\begin{equation}
    \delta_{xx^*}\delta_{ae}\rho_\rA -X^{a|x}_\rA +Y^{e|x}_\rA -\delta_{xx^*}\sum_{x'\in\cX}Y^{e|x'}_\rA +Z^{ae|x}_\rA \neq \zero_\rA\;,
\end{equation}
one can make $\sup_{G^{ae|x}_\rA}L(X^{a|x}_\rA ,Y^{e|x}_\rA ,Z^{ae|x}_\rA,G^{ae|x}_\rA)$ infinitely large.
Therefore, to have a non-trivial value, it must hold that
\begin{equation}\label{eq:dual_derivation}
    \delta_{xx^*}\delta_{ae}\rho_\rA -X^{a|x}_\rA +Y^{e|x}_\rA -\delta_{xx^*}\sum_{x'\in\cX}Y^{e|x'}_\rA +Z^{ae|x}_\rA =\zero_\rA.
\end{equation}
Note that now we set $Z^{ae|x}_\rA\succeq\zero_\rA$.
Combining this and \eq{eq:dual_derivation} leads to the constraint
\begin{equation}
    Z^{ae|x}_\rA=-\delta_{xx^*}\delta_{ae}\rho_\rA+X^{a|x}_\rA -Y^{e|x}_\rA+\delta_{xx^*}\sum_{x'\in\cX}Y^{e|x'}_\rA\succeq\zero_\rA.
\end{equation}
Finally, we get the dual problem.
\begin{equation}\label{eq:dual_guessing_PMD}
    \begin{split}
        \given \quad &\mbM^{\cA|\cX}_\rA,\rho_\rA\;,\\
        \minimize{X^{a|x}_\rA,Y^{e|x}_\rA}\quad&\sum_{x\in\cX}\sum_{a\in\cA}\Tr[X^{a|x}_\rA  M^{a|x}_\rA]\\
        \subto \quad&X^{a|x}_\rA -Y^{e|x}_\rA+\delta_{xx^*}\sum_{x'\in\cX}Y^{e|x'}_\rA\succeq\delta_{xx^*}\delta_{ae}\rho_\rA \quad(\forall x\in\cX,\;\forall a,e\in\cA).
    \end{split}
\end{equation}

We see that the strong duality holds.
Indeed, by choosing $X^{a|x}_\rA=\alpha\one_\rA$ where $\alpha>1$ and $Y^{e|x}_\rA=\zero_\rA$, they are strict feasible solutions because it holds that $X^{a|x}_\rA\succ\delta_{xx^*}\delta_{ae}\rho_\rA$.
Also, the optimal solution of the primal problem Eq.~\eqref{eq:guessing_PMD} is at most $1$, which is finite.
Thus, from Slater's theorem~\cite{slater2014lagrange,boyd2004convex}, the optimal value $p^{x^*}_\guess (\mbM^{\cA|\cX}_\rA,\rho_\rA)$ of the primal problem Eq.~\eqref{eq:guessing_PMD} is equal to the optimal value of the dual problem Eq.~\eqref{eq:dual_guessing_PMD}, and the primal problem Eq.~\eqref{eq:guessing_PMD} has optimal solutions because the optimal value is finite by definition of the problem.

\subsection{Proof of Proposition~\ref{proposition:assemblage_PMD}}

\begin{proof}[Proof of Proposition~\ref{proposition:assemblage_PMD}]
    To show this inequality, it is enough to show that, for any assemblage $\sigma(\mbM^{\cA|\cX}_\rA,\rho_{\rA\rB})$ such that $\Tr_\rB\rho_{\rA\rB}=\rho_\rA$, it holds that
    \begin{equation}
        p_\guess^{x^*}(\sigma(\mbM^{\cA|\cX}_\rA,\rho_{\rA\rB}))\ge p^{x^*}_\guess (\mbM^{\cA|\cX}_\rA,\rho_\rA).
    \end{equation}
    As confirmed in the previous subsection, Eq.~\eqref{eq:guessing_PMD} has an optimal solution.
    Then, take an optimal solution and denote it $\{G^{\circ ae|x}_\rA\}_{a,e\in\cA,x\in\cX}$ satisfying
    \begin{align}
        \sum_{e\in\cA}G^{\circ ae|x}_\rA&=M^{a|x}_\rA\quad(\forall x\in\cX,\forall a\in\cA)\\
        \sum_{a\in\cA}G^{\circ ae|x}_\rA&=\sum_{b\in\cA}G^{\circ be|x^*}_\rA\quad(\forall x\in\cX,\forall e\in\cA)
    \end{align}
    Take a state $\rho_{\rA\rB}\in\cS(\cH_\rA\otimes\cH_\rB)$ such that $\Tr_\rB\rho_{\rA\rB}=\rho_\rA$ and define
    \begin{equation}
        \tilde{\tau}_\rB^{ae|x}:=\Tr_\rA[(G^{\circ ae|x}_\rA\otimes\one_\rB)\rho_{\rA\rB}].
    \end{equation}
    Now we confirm that $\tilde{\tau}_\rB^{ae|x}$ is a feasible solution of Eq.~\eqref{eq:guessing_assemblage}.
    First, for all $x\in\cX,a\in\cA$, it holds that 
    \begin{equation}
        \begin{split}
            \sum_{e\in\cA}\tilde{\tau}_\rB^{ae|x}&=\sum_{e\in\cA}\Tr_\rA[(G^{\circ ae|x}_\rA\otimes\one_\rB)\rho_{\rA\rB}]\\
            &=\Tr_\rA\left[\left(\sum_{e\in\cA}G^{\circ ae|x}_\rA\otimes\one_\rB\right)\rho_{\rA\rB}\right]\\
            &=\Tr_\rA\left[\left(M^{a|x}_\rA\otimes\one_\rB\right)\rho_{\rA\rB}\right]=\sigma _\rB ^{a|x}.
        \end{split}
    \end{equation}
    Thus, the first constraint of Eq.~\eqref{eq:guessing_assemblage} is satisfied.
    
    Next, for all $x\in\cX$ and for all $e\in\cA$, it holds that
    \begin{equation}
        \begin{split}
            \sum_{a\in\cA}\tilde{\tau}_\rB^{ae|x}&=\sum_{a\in\cA}\Tr_\rA[(G^{\circ ae|x}_\rA\otimes\one_\rB)\rho_{\rA\rB}]\\
            &=\Tr_\rA\left[\left(\sum_{a\in\cA}G^{\circ ae|x}_\rA\otimes\one_\rB\right)\rho_{\rA\rB}\right]\\
            &=\Tr_\rA\left[\left(\sum_{b\in\cA}G^{\circ be|x^*}_\rA\otimes\one_\rB\right)\rho_{\rA\rB}\right]\\
            &=\sum_{b\in\cA}\Tr_\rA\left[\left(G^{\circ be|x^*}_\rA\otimes\one_\rB\right)\rho_{\rA\rB}\right]\\
            &=\sum_{b\in\cA}\tilde{\tau}^{be|x^*}_\rB.
        \end{split}
    \end{equation}
    Therefore, the second constraint is satisfied.
    
    Let us check the third constraint, positive-semidefiniteness of $\tilde{\tau}_\rB^{a_xe|x}$.
    We can write the assemblage like
    \begin{equation}
        \begin{split}
            \tilde{\tau}_\rB^{ae|x}&=\Tr_\rA[(G^{\circ ae|x}_\rA\otimes\one_\rB)\rho_{\rA\rB}]\\
            &=\Tr_\rA\left[\left(\sqrt{G^{\circ ae|x}_\rA}\otimes\one_\rB\right)\rho_{\rA\rB}\left(\sqrt{G^{\circ ae|x}_\rA}\otimes\one_\rB\right)\right]
        \end{split}
    \end{equation}
    by using Eq.~\eqref{eq:ptrace_cyclicity}.
    Notice that
    \begin{equation}
        \left(\sqrt{G^{\circ ae|x}_\rA}\otimes\one_\rB\right)\rho_{\rA\rB}\left(\sqrt{G^{\circ ae|x}_\rA}\otimes\one_\rB\right)\succeq\zero_{\rA\rB}
    \end{equation}
    holds, and the partial trace preserves the positive semidefiniteness.
    Therefore, $\tilde{\tau}_\rB^{ae|x}\succeq\zero$.

    So far, we have confirmed that one can construct a feasible solution of Eq.~\eqref{eq:guessing_assemblage} from the optimal solution of Eq.~\eqref{eq:guessing_PMD}.
    In addition to this, it holds that
    \begin{equation}
        \begin{split}
            \sum_{e\in\cA}\Tr\tilde{\tau}^{ee|x^*}_\rB&=\sum_{e\in\cA}\Tr\left[\left(G^{\circ  ee|x^*}_\rA\otimes\one_\rB\right)\rho_{\rA\rB}\right]\\
            &=\sum_{e\in\cA}\Tr\left[G^{\circ  ee|x^*}_\rA\rho_\rA\right].
        \end{split}
    \end{equation}
    This means that the optimal value of Eq.(7) is the same as the value of Eq.~\eqref{eq:guessing_assemblage} for this feasible solution $\tilde{\tau}^{a_xe|x}_\rB$.
    However, since $\tilde{\sigma}_\rB^{a_xe}$ is not necessarily the optimal solution of Eq.~\eqref{eq:guessing_assemblage}, the desired inequality holds.
\end{proof}

\subsection{Proof of Theorem~\ref{theorem:assemblage_PMD}}

Before proving Theorem~\ref{theorem:assemblage_PMD}, we give a Lipschitz continuity of $p^{x^*}_\guess (\mbM^{\cA|\cX}_\rA,\rho_\rA)$.
\begin{lemma}\label{lemma:lipschitz}
    For a PMD $\mbM^{\cA|\cX}_\rA$ and states $\rho_\rA,\rho'_\rA\in\cS(\cH_\rA)$, the following inequality holds.
    \begin{equation}
        |p^{x^*}_\guess (\mbM^{\cA|\cX}_\rA,\rho_\rA)-p^{x^*}_\guess (\mbM^{\cA|\cX}_\rA,\rho'_\rA)|\le\|\rho_\rA-\rho'_\rA\|_1.
    \end{equation}
\end{lemma}
\begin{proof}
    Take an optimal solution for $p^{x^*}_\guess (\mbM^{\cA|\cX}_\rA,\rho_\rA)$, $G^{\circ ae|x}_\rA$, which satisfies $\zero_\rA\preceq\sum_{e\in\cA}G^{\circ ee|x^*}_\rA\preceq\one_\rA$.
    Since $G^{\circ ae|x}_\rA$ is not necessarily optimal for $p^{x^*}_\guess (\mbM^{\cA|\cX}_\rA,\rho'_\rA)$, by using an expression of trace norm $\|\rho_\rA-\rho'_\rA\|_1=\max_{\|V_\rA\|\le1}|\Tr[V_\rA(\rho_\rA-\rho'_\rA)]|$ (see e.g., \cite[Section 1.1.3]{watrous2018theory}), one has
    \begin{equation}
    \begin{split}
        p^{x^*}_\guess (\mbM^{\cA|\cX}_\rA,\rho_\rA)-p^{x^*}_\guess (\mbM^{\cA|\cX}_\rA,\rho'_\rA)&\le \Tr\left[\sum_{e\in\cA}G^{\circ ee|x^*}_\rA(\rho_\rA-\rho'_\rA)\right]\\
        &\le \|\rho_\rA-\rho'_\rA\|_1.
    \end{split}        
    \end{equation}
    Conversely, by taking the optimal solution of $p^{x^*}_\guess (\mbM^{\cA|\cX}_\rA,\rho'_\rA)$, $G'^{\circ ae|x}_\rA$, one has 
    \begin{equation}
        \begin{split}
            p^{x^*}_\guess (\mbM^{\cA|\cX}_\rA,\rho'_\rA)-p^{x^*}_\guess (\mbM^{\cA|\cX}_\rA,\rho_\rA)&\le \Tr\left[\sum_{e\in\cA}G'^{\circ ee|x^*}_\rA(\rho'_\rA-\rho_\rA)\right]\\
            &\le \|\rho'_\rA-\rho_\rA\|_1\\
            &= \|\rho_\rA-\rho'_\rA\|_1.
        \end{split}
    \end{equation}
    Combining the above inequalities yields the desired inequality.
\end{proof}

\begin{proof}[Proof of Theorem~\ref{theorem:assemblage_PMD}]
    First, due to Proposition~\ref{proposition:assemblage_PMD}, it holds that
    \begin{equation}\label{eq:1SDI>PMD}
        \inf_{\rho_{\rA\rB}\in\cS(\cH_\rA\otimes\cH_\rB)} p_\guess^{x^*}(\sigma(\mbM^{\cA|\cX}_\rA,\rho_{\rA\rB}))\ge\min_{\rho_\rA\in\cS(\cH_\rA)}p^{x^*}_\guess (\mbM^{\cA|\cX}_\rA,\rho_\rA).
    \end{equation}
    
    Next, under the assumption of $d=\dim\cH_\rA\le\dim\cH_\rB$, we show
    \begin{equation}\label{eq:1SDI<PMD}
        \inf_{\rho_{\rA\rB}\in\cS(\cH_\rA\otimes\cH_\rB)}p_\guess^{x^*}(\sigma(\mbM^{\cA|\cX}_\rA,\rho_{\rA\rB}))\le \min_{\rho_\rA\in\cS(\cH_\rA)}p^{x^*}_\guess (\mbM^{\cA|\cX}_\rA,\rho_\rA),
    \end{equation}
    To show this, first introduce a canonical purification of $\rho_\rA$, $\left(\sqrt{\rho_\rA}\otimes\one_\rB\right)\sum_{i=0}^{d-1}\ket{i}_\rA\otimes\ket{i}_{\rA'}$.
    Here, $\{\ket{i}_\rA\}_{i=0}^{d-1}$ is an orthonormal basis of $\cH_\rA$ that diagonalizes $\rho_\rA$ and $\{\ket{i}_{\rA'}\}_{i=0}^{d-1}$ is an orthonormal basis of $\cH_{\rA'}$, which is a $d$-dimensional subset of $\cH_\rB$ (thus, $\cH_{\rA'}$ is isomorphic to $\cH_\rA$).
    Then, one take an isometry $V_{\rA\to\rB}:\cH_\rA\to\cH_\rB$ such that $V\ket{i}_\rA=\ket{i}_\rB$.
    Define
    \begin{equation}
        \ket{\rho}_{\rA\rB}:=(\sqrt{\rho_\rA}\otimes V_{\rA\to\rB})\ket{\one}_{\rA\rA'}
    \end{equation}
    where $\ket{\one}_{\rA\rA'}:=\sum_{i}\ket{i}_\rA\otimes\ket{i}_{\rA'}$ is the unnormalized maximally entangled state.

    Our aim is now to show the following inequality when $\rho_\rA$ is a full-rank state:
    \begin{equation}\label{eq:purification}
        p_\guess^{x^*}(\sigma(\mbM^{\cA|\cX}_\rA,\ket{\rho}_{\rA\rB}))\le p^{x^*}_\guess (\mbM^{\cA|\cX}_\rA,\rho_\rA).
    \end{equation}
    In this case, the assemblage $\sigma(\mbM^{\cA|\cX}_\rA,\ket{\rho}_{\rA\rB})$ consists of
    \begin{equation}
        \begin{split}
            \sigma(\mbM^{\cA|\cX}_\rA,\ket{\rho}_{\rA\rB})&=\Tr_\rA\left[(M^{a|x}_\rA\otimes\one_\rB)(\sqrt{\rho_\rA}\otimes V_{\rA\to\rB})\ketbra{\one}{\one}_{\rA\rA'}(\sqrt{\rho_\rA}\otimes V^\dagger_{\rA\to\rB})\right]\\
            &=\Tr_\rA\left[(\one_\rA\otimes V_{\rA\to\rB})(M^{a|x}_\rA\sqrt{\rho_\rA}\otimes\one_{\rA'})\ketbra{\one}{\one}_{\rA\rA'}(\sqrt{\rho_\rA}\otimes\one_{\rA'})(\one_\rA\otimes V^\dagger_{\rA\to\rB})\right]\\
            &=\Tr_\rA\left[(\one_\rA\otimes V_{\rA\to\rB}\sqrt{\rho_\rA}(M^{a|x}_\rA)^\top)\ketbra{\one}{\one}_{\rA\rA'}(\one_\rA\otimes \sqrt{\rho_\rA}V^\dagger_{\rA\to\rB})\right]\\
            &=V_{\rA\to\rB}\sqrt{\rho_\rA}(M^{a|x}_\rA)^\top\sqrt{\rho_\rA}V^\dagger_{\rA\to\rB}.
        \end{split}
    \end{equation}
    Take $\eta^{ae|x}_\rB$ as a feasible solution of Eq.~\eqref{eq:guessing_assemblage} for the assemblage $\sigma(\mbM^{\cA|\cX}_\rA,\ket{\rho}_{\rA\rB})$, that is, $\eta^{ae|x}_\rB$ satisfies
    \begin{align}
        \sum_{e\in\cA}\eta^{ae|x}_\rB&=V_{\rA\to\rB}\sqrt{\rho_\rA}(M^{a|x}_\rA)^\top\sqrt{\rho_\rA}V^\dagger_{\rA\to\rB}\quad(\forall x\in\cX,\forall a\in\cA)\label{eq:eta_marginal_e}\\
        \sum_{a\in\cA}\eta^{ae|x}_\rB&=\sum_{b\in\cA}\eta^{be|x^*}_\rB\quad(\forall x\in\cX,\forall e\in\cA)\\
        \eta^{ae|x}_\rB&\succeq\zero_\rB\quad(\forall x\in\cX,\forall a,e\in\cA)
    \end{align}
    Noticing that
    \begin{equation}\label{eq:isometry}
        \sum_{a,e\in\cA}\eta^{ae|x}_\rB=V_{\rA\to\rB}\rho_\rA V^\dagger_{\rA\to\rB}
    \end{equation}
    holds for all $x\in\cX$, define $H^{ae|x}_\rA\in\cL(\cH_\rA)$ as
    \begin{equation}\label{eq:H}
        H^{ae|x}_\rA:=\frac{1}{\sqrt{\rho_\rA}}(V^\dagger_{\rA\to\rB}\eta^{ae|x}_{\rB}V_{\rA\to\rB})^\top\frac{1}{\sqrt{\rho_\rA}}.
    \end{equation}
    Clearly $H^{ae|x}_\rA\succeq\zero_\rA$, and we have
    \begin{equation}
        \sum_{a,e\in\cA}H^{ae|x}_\rA=\frac{1}{\sqrt{\rho_\rA}}\left(V^\dagger_{\rA\to\rB}\left(\sum_{a,e\in\cA}\eta^{ae|x}_{\rB}\right)V_{\rA\to\rB}\right)^\top\frac{1}{\sqrt{\rho_\rA}}=\one_\rA,
    \end{equation}
    which mean that $\{H^{ae|x}_\rA\}_{a,e\in\cA}$ is a POVM.
    Moreover, due to Eq.~\eqref{eq:eta_marginal_e}, we have
    \begin{equation}
        \sum_{e\in\cA}H^{ae|x}_\rA=\frac{1}{\sqrt{\rho_\rA}}\left(V^\dagger_{\rA\to\rB}\left(\sum_{e\in\cA}\eta^{ae|x}_{\rB}\right)V_{\rA\to\rB}\right)^\top\frac{1}{\sqrt{\rho_\rA}}=M^{a|x}_\rA
    \end{equation}
    and
    \begin{equation}
        \begin{split}
            \sum_{a\in\cA}H^{ae|x}_\rA&=\frac{1}{\sqrt{\rho_\rA}}\left(V^\dagger_{\rA\to\rB}\left(\sum_{a\in\cA}\eta^{ae|x}_{\rB}\right)V_{\rA\to\rB}\right)^\top\frac{1}{\sqrt{\rho_\rA}}=\frac{1}{\sqrt{\rho_\rA}}\left(V^\dagger_{\rA\to\rB}\left(\sum_{b\in\cA}\eta^{be|x^*}_{\rB}\right)V_{\rA\to\rB}\right)^\top\frac{1}{\sqrt{\rho_\rA}}\\
            &=\sum_{b\in\cA}H^{be|x^*}_\rA
        \end{split}
    \end{equation}
    Therefore, $H^{ae|x}_\rA$ defined in Eq.~\eqref{eq:H} is a feasible solution of Eq.~\eqref{eq:guessing_PMD}.
    Then, the objective function of Eq.~\eqref{eq:guessing_assemblage} is equal to Eq.~\eqref{eq:guessing_PMD}.
    \begin{equation}
        \sum_{e\in\cA}\Tr[H^{ee|x^*}_\rA\rho_\rA]=\sum_{e\in\cA}\Tr\left[\frac{1}{\sqrt{\rho_\rA}}(V^\dagger_{\rA\to\rB}\eta^{ee|x^*}_{\rB}V_{\rA\to\rB})^\top\frac{1}{\sqrt{\rho_\rA}}\rho_\rA\right]=\sum_{e\in\cA}\Tr\left[V_{\rA\to\rB}V^\dagger_{\rA\to\rB}\eta^{ee|x^*}_{\rB}\right]=\sum_{e\in\cA}\Tr\eta^{ee|x^*}_{\rB}.
    \end{equation}
    The last equality used $V_{\rA\to\rB}V^\dagger_{\rA\to\rB}\eta^{ee|x^*}_{\rB}=\eta^{ee|x^*}_{\rB}$ since we have $\supp \eta^{ee|x^*}_{\rB}\subseteq\mathrm{ran}\;V$ as shown in the following.
    From Eq.~\eqref{eq:isometry}, we have $\zero_\rB\preceq \eta^{ae|x}_\rB\preceq V_{\rA\to\rB}\rho_\rA V^\dagger_{\rA\to\rB}$, which implies $\ker V_{\rA\to\rB}\rho_\rA V^\dagger_{\rA\to\rB}\subseteq\ker\eta^{ae|x}_\rB$ and thus, we have
    \begin{align}
        \supp \eta^{ae|x}_\rB&\subseteq\supp V_{\rA\to\rB}\rho_\rA V^\dagger_{\rA\to\rB}\\
        &=(\ker V_{\rA\to\rB}\rho_\rA V^\dagger_{\rA\to\rB})^\perp\\
        &=\mathrm{ran}(V_{\rA\to\rB}\rho_\rA V^\dagger_{\rA\to\rB})^\dagger\\
        &=\mathrm{ran}\;V_{\rA\to\rB}\rho_\rA V^\dagger_{\rA\to\rB}\\
        &\subseteq \mathrm{ran}\;V_{\rA\to\rB}.
    \end{align}
    Since $\eta^{ae|x}_\rB\succeq\zero_\rB$, one has $\supp \eta^{ae|x}_\rB=\mathrm{ran}\;\eta^{ae|x}_\rB$.
    In summary, we have $\mathrm{ran}\;\eta^{ae|x}_\rB\subseteq\mathrm{ran}\;V$.
    Now, take an arbitrary vector $\ket{\psi}_\rB\in\cH_\rB$.
    We have $V_{\rA\to\rB}V^\dagger_{\rA\to\rB}\eta^{ae|x}_\rB\ket{\psi}_\rB=\eta^{ae|x}_\rB\ket{\psi}_\rB$ because $V_{\rA\to\rB}V^\dagger_{\rA\to\rB}$ does not change any vector in $\mathrm{ran}\;V_{\rA\to\rB}$.
    Thus, $V_{\rA\to\rB}V^\dagger_{\rA\to\rB}\eta^{ae|x}_\rB=\eta^{ae|x}_\rB$.
    Finally, $H^{ae|x}_\rA$ is not necessarily optimal solution, which implies Eq.~\eqref{eq:purification}.

    From the above discussion, if the optimal of $\min_{\rho_\rA\in\cS(\cH_\rA)}p^{x^*}_\guess (\mbM^{\cA|\cX}_\rA,\rho_\rA)$ is achieved by a full rank state, it holds that
    \begin{equation}
        \min_{\rho_{\rA\rB}\in\cS(\cH_\rA\otimes\cH_\rB)}p_\guess^{x^*}(\sigma(\mbM^{\cA|\cX}_\rA,\rho_{\rA\rB}))\le \min_{\rho_\rA\in\cS(\cH_\rA)}p^{x^*}_\guess (\mbM^{\cA|\cX}_\rA,\rho_\rA),
    \end{equation}
    If the optimal solution of $\min_{\rho_\rA\in\cS(\cH_\rA)}p^{x^*}_\guess (\mbM^{\cA|\cX}_\rA,\rho_\rA)$ is not full rank state $\rho^\circ_\rA$, introduce
    \begin{equation}
        \rho^{(\epsilon)}_\rA:=(1-\epsilon)\rho^\circ_\rA+\epsilon\frac{\one_\rA}{d}\quad(0<\epsilon\le 1).
    \end{equation}
    Since $\rho^{(\epsilon)}_\rA$ is a full rank state, one can use Eq.~\eqref{eq:purification} and it holds that
    \begin{equation}
        p_\guess^{x^*}(\sigma(\mbM^{\cA|\cX}_\rA,\ket{\rho^{(\epsilon)}}_{\rA\rB}))\le p^{x^*}_\guess (\mbM^{\cA|\cX}_\rA,\rho^{(\epsilon)}_\rA)\quad(0<\epsilon\le 1).
    \end{equation}
    Therefore, since we have
    \begin{equation}
        \inf_{\rho_{\rA\rB}\in\cS(\cH_\rA\otimes\cH_\rB)}p_\guess^{x^*}(\sigma(\mbM^{\cA|\cX}_\rA,\rho_{\rA\rB}))\le p_\guess^{x^*}(\sigma(\mbM^{\cA|\cX}_\rA,\ket{\rho^{(\epsilon)}}_{\rA\rB}))\le p^{x^*}_\guess (\mbM^{\cA|\cX}_\rA,\rho^{(\epsilon)}_\rA)\quad(0<\epsilon\le 1).
    \end{equation}
    and from Lemma~\ref{lemma:lipschitz}, for all $\epsilon$ such that $0<\epsilon\le 1$, it holds that
    \begin{equation}
        |p^{x^*}_\guess (\mbM^{\cA|\cX}_\rA,\rho^{(\epsilon)}_\rA)-p^{x^*}_\guess (\mbM^{\cA|\cX}_\rA,\rho^\circ_\rA)|\le \|\rho^{(\epsilon)}_\rA-\rho^\circ_\rA\|_1,
    \end{equation}
    which implies
    \begin{equation}
        \lim_{\epsilon\to0}p^{x^*}_\guess (\mbM^{\cA|\cX}_\rA,\rho^{(\epsilon)}_\rA)=p^{x^*}_\guess (\mbM^{\cA|\cX}_\rA,\rho^\circ_\rA)=\min_{\rho_\rA\in\cS(\cH_\rA)}p^{x^*}_\guess (\mbM^{\cA|\cX}_\rA,\rho_\rA).
    \end{equation}   
    Therefore, Eq.~\eqref{eq:1SDI<PMD} holds regardless of whether the optimal state of $\min_{\rho_\rA\in\cS(\cH_\rA)}p^{x^*}_\guess (\mbM^{\cA|\cX}_\rA,\rho_\rA)$ is full rank or not.

    Combining Eqs.~\eqref{eq:1SDI>PMD} and \eqref{eq:1SDI<PMD} yields the desired equality.
\end{proof}

\section{Star-incompatibility estimation from randomness}

\subsection{Semidefinite programming of star-incompatibility weight}
Define $s:=1-w$ in Definition 4.
Then we consider the following problem:
\begin{equation}
    \begin{split}
        \given \quad &\mbM^{\cA|\cX}_\rA\quad \\
        \maximize{s,F^{a|x}_\rA,E^{a|x}_\rA}\quad &s\quad(0\le s\le 1)\\
        \subto \quad & M^{a|x}_\rA=sF^{a|x}_\rA+(1-s)E^{a|x}_\rA\;(\forall x\in\cX,\forall a\in\cA),\\
        &\F^{\cA|x}_\rA,\E^{\cA|x}_\rA\in\sM^{\cA}_\rA\quad(\forall x\in\cX)\\
            \quad & \{\F^{\cA|x}_\rA\}_{x\in\cX}\text{ is star-compatible}.
    \end{split}
\end{equation}
and we have $W^{x^*}(\mbM^{\cA|\cX}_\rA)=1-s^*$ where $s^*$ is the optimal value of this problem.

Introduce $\tilde{F}^{a|x}_\rA:=sF^{a|x}_\rA$.
Then, when $s<1$ (the weight has a positive value), the problem becomes
\begin{equation}
    \begin{split}
        \given \quad &\mbM^{\cA|\cX}_\rA\quad \\
        \maximize{s\in\mathbb R,\tilde{F}^{a|x}_\rA,E^{a|x}_\rA,\tilde{G}^{ae|x}_\rA}\quad &s\\
        \subto \quad & E^{a|x}_\rA=\frac{1}{1-s}\left(M^{a|x}_\rA-\tilde{F}^{a|x}_\rA\right)\quad(\forall x\in\cX,\forall a\in\cA),\\
        &E^{a|x}_\rA\succeq\zero_\rA\quad(\forall x\in\cX,\forall a\in\cA),\\
        &\sum_{a\in\cA}E^{a|x}_\rA=\one_\rA\quad(\forall x\in\cX),\\
        &\sum_{e\in\cA}\tilde{G}^{ae|x}_\rA=\tilde{F}^{a|x}_\rA\quad(\forall x\in\cX,\forall a\in\cA),\\
        &\sum_{a\in\cA}\tilde{G}^{ae|x}_\rA=\tilde{F}^{e|x^*}_\rA\quad(\forall x\in\cX,\forall e\in\cA),\\
        &\sum_{a\in\cA}\tilde{F}^{a|x}_\rA=s\one_\rA\quad(\forall x\in\cX)\\
        &\tilde{F}^{a|x}_\rA\succeq\zero_\rA\quad(\forall x\in\cX,a\in\cA)\\
        &\tilde{G}^{ae|x}_\rA\succeq\zero_\rA\quad(\forall a,e\in\cA).
    \end{split}
\end{equation}
The first and second constraints can be combined and become $M^{a|x}_\rA-\tilde{F}^{a|x}_\rA\succeq\zero_\rA$.
We can remove the third constraint due to the sixth constraint.
The seventh constraint is derived from the fourth, fifth, and eighth constraints.
Thus, the problem becomes the following SDP (although the case $s=1$ has been excluded from the derivation, it is included in the following formulation).
\begin{equation}\label{eq:1-weight}
    \begin{split}
        \given \quad &\mbM^{\cA|\cX}_\rA\quad \\
        \maximize{s\in\mathbb R,\tilde{F}^{a|x}_\rA,\tilde{G}^{ae|x}_\rA}\quad &s\\
        \subto \quad&\sum_{e\in\cA}\tilde{G}^{ae|x}_\rA=\tilde{F}^{a|x}_\rA\quad(\forall x\in\cX,\forall a\in\cA),\\        \quad&\sum_{a\in\cA}\tilde{G}^{ae|x}_\rA=\tilde{F}^{e|x^*}_\rA\quad(\forall x\in\cX,\forall e\in\cA),\\
        \quad & \sum_{a\in\cA}\tilde{F}^{a|x}_\rA=s\one_\rA\quad(\forall x\in\cX)\\
        \quad & M^{a|x}_\rA-\tilde{F}^{a|x}_\rA\succeq\zero_\rA\quad(\forall x\in\cX,\forall a\in\cA),\\
        \quad&\tilde{G}^{ae|x}_\rA\succeq\zero_\rA\quad(\forall x\in\cX,\forall a,e\in\cA).
    \end{split}
\end{equation}
Let us derive the dual problem.
Define a Lagrangian
\begin{equation}
    \begin{split}
        &L(X^{a|x}_\rA,Y^{e|x}_\rA,Z_\rA,R^{a|x}_\rA,S^{ae|x}_\rA;s,\tilde{F}^{a|x}_\rA,\tilde{G}^{ae|x}_\rA)\\
        &:=
        \begin{dcases}
            s+\sum_{x\in\cX}\sum_{a\in\cA}\Tr\left[X^{a|x}_\rA\left(\tilde{F}^{a|x}_\rA-\sum_{e\in\cA}\tilde{G}^{ae|x}_\rA\right)\right]+\sum_{x\in\cX}\sum_{e\in\cA}\Tr\left[Y^{e|x}_\rA\left(\tilde{F}^{e|x^*}_\rA-\sum_{a\in\cA}\tilde{G}^{ae|x}_\rA\right)\right]\\
            +\sum_{x\in\cX}\Tr\left[Z^x_\rA\left(\sum_{a\in\cA}\tilde{F}^{a|x}_\rA-s\one_\rA\right)\right]+\sum_{x\in\cX}\sum_{a\in\cA}\Tr\left[R^{a|x}_\rA\left(M^{a|x}_\rA-\tilde{F}^{a|x}_\rA\right)\right]\\
            +\sum_{x\in\cX}\sum_{a,e\in\cA}\Tr[S^{ae|x}_\rA\tilde{G}^{ae|x}_\rA]\quad(R^{a|x}_\rA,S^{ae|x}_\rA\succeq\zero_\rA),\\
            +\infty\quad(\text{otherwise}),
        \end{dcases}
    \end{split}
\end{equation}
where $X^{a|x}_\rA,Y^{e|x}_\rA,Z_\rA,R^{a|x}_\rA,S^{ae|x}_\rA,\tilde{F}^{a|x}_\rA,\tilde{G}^{ae|x}_\rA\in\cL_\rH(\cH_\rA)$ and $s\in\mathbb R$.
Then, the primal problem Eq.~\eqref{eq:1-weight} is
\begin{equation}
    \sup_{s\in\mathbb R,\tilde{F}^{a|x}_\rA,\tilde{G}^{ae|x}_\rA\in\cL_\rH(\cH_\rA)}\inf_{X^{a|x}_\rA,Y^{e|x}_\rA,Z_\rA,R^{a|x}_\rA,S^{ae|x}_\rA\in\cL_\rH(\cH_\rA)}L(X^{a|x}_\rA,Y^{e|x}_\rA,Z_\rA,R^{a|x}_\rA,S^{ae|x}_\rA;s,\tilde{F}^{a|x}_\rA,\tilde{G}^{ae|x}_\rA).
\end{equation}
The Lagrangian for $R^{a|x}_\rA,S^{ae|x}_\rA\succeq\zero_\rA$ can be rewritten as follows.
\begin{equation}
    \begin{split}
        &L(X^{a|x}_\rA,Y^{e|x}_\rA,Z_\rA,R^{a|x}_\rA,S^{ae|x}_\rA;s,\tilde{F}^{a|x}_\rA,\tilde{G}^{ae|x}_\rA)\\
        &=\sum_{x\in\cX}\sum_{a\in\cA}\Tr[R^{a|x}_\rA M^{a|x}_\rA]\\
        &+s\left(1-\sum_{x\in\cX}\Tr Z^x_\rA\right)\\
        &+\sum_{x\in\cX}\sum_{a\in\cA}\Tr\left[\left(X^{a|x}_\rA+Z^x_\rA-R^{a|x}_\rA+\delta_{xx^*}\sum_{x'\in\cX}Y^{a|x'}_\rA\right)\tilde{F}^{a|x}_\rA\right]\\
        &+\sum_{x\in\cX}\sum_{a,e\in\cA}\Tr[(S^{ae|x}_\rA-Y^{e|x}_\rA-X^{a|x}_\rA)\tilde{G}^{ae|x}_\rA].
    \end{split}
\end{equation}
To avoid $\sup_{s\in\mathbb R,\tilde{F}^{a|x}_\rA,\tilde{G}^{ae|x}_\rA\in\cL_\rH(\cH_\rA)}L(X^{a|x}_\rA,Y^{e|x}_\rA,Z_\rA,R^{a|x}_\rA,S^{ae|x}_\rA;s,\tilde{F}^{a|x}_\rA,\tilde{G}^{ae|x}_\rA)=+\infty$, we need
\begin{align}
    1-\sum_{x\in\cX}\Tr Z^x_\rA&=0,\\
    X^{a|x}_\rA+Z^x_\rA-R^{a|x}_\rA+\delta_{xx^*}\sum_{x'\in\cX}Y^{a|x'}_\rA&=\zero_\rA,\\
    S^{ae|x}_\rA-Y^{e|x}_\rA-X^{a|x}_\rA&=\zero_\rA.
\end{align}
Thus, the dual problem
\begin{equation}
    \inf_{X^{a|x}_\rA,Y^{e|x}_\rA,Z_\rA,R^{a|x}_\rA,S^{ae|x}_\rA\in\cL_\rH(\cH_\rA)}\sup_{s\in\mathbb R,\tilde{F}^{a|x}_\rA,\tilde{G}^{ae|x}_\rA\in\cL_\rH(\cH_\rA)}L(X^{a|x}_\rA,Y^{e|x}_\rA,Z_\rA,R^{a|x}_\rA,S^{ae|x}_\rA;s,\tilde{F}^{a|x}_\rA,\tilde{G}^{ae|x}_\rA)
\end{equation}
is written as
\begin{equation}\label{eq:1-weight_dual}
    \begin{split}
        \given\quad &\mbM^{\cA|\cX}_\rA\\
        \minimize{R^{a|x}_\rA,X^{a|x}_\rA,Y^{e|x}_\rA,Z^x_\rA\in\cL_{\rH}(\cH_\rA)}\quad &\sum_{x\in\cX}\sum_{a\in\cA}\Tr[R^{a|x}_\rA M^{a|x}_\rA]\\
        \subto\quad & \sum_{x\in\cX}\Tr Z^x_\rA=1,\\
        \quad & X^{a|x}_\rA+Z^x_\rA+\delta_{xx^*}\sum_{x'\in\cX}Y^{a|x'}_\rA=R^{a|x}_\rA\quad(\forall x\in\cX,\forall a\in\cA),\\
        \quad &X^{a|x}_\rA+Y^{e|x}_\rA\succeq\zero_\rA\quad(\forall x\in\cX,\forall a,e\in\cA),\\
        \quad & R^{a|x}_\rA\succeq\zero_\rA\quad(\forall x\in\cX,\forall a\in\cA).
    \end{split}
\end{equation}

One can construct strictly feasible solutions of this problem like $X^{a|x}_\rA=\alpha\one_\rA(\alpha>\frac{1}{d_\rA|\cX|}),Y^{a|x}_\rA=\zero_\rA,Z^x_\rA=\frac{\one_\rA}{d_\rA|\cX|}$.
Thus, from Slater's theorem~\cite{slater2014lagrange,boyd2004convex}, the primal problem Eq.~\eqref{eq:1-weight} and the dual problem Eq.~\eqref{eq:1-weight_dual} have the same optimal value, and the primal problem has optimal solutions because the optimal value is finite by definition of the problem.

\subsection{Proof of Theorem~\ref{theorem:weight}}
\begin{proof}[Proof of Theorem~\ref{theorem:weight}]
    Take optimal measurements of the decomposition of Eq.~\eqref{eq:weight} as $\F^{\circ \cA|x}_\rA,\E^{\circ\cA|x}_\rA\in\sM^\cA_\rA$ and denote $\mbF^{\circ\cA|\cX}_\rA:=\{\F^{\circ \cA|x}_\rA\}_{x\in\cX}$ and $\mbE^{\circ\cA|\cX}_\rA:=\{\E^{\circ \cA|x}_\rA\}_{x\in\cX}$.
    Here, $\mbF^{\circ\cA|\cX}_\rA:=\{\F^{\circ \cA|x}_\rA\}_{x\in\cX}$ is star-compatible.
    Then, there are POVMs $\{\J^{\cA\times\cA|x}_\rA\}_{x\in\cX\setminus\{x^*\}}$ where $\J^{\cA\times\cA|x}_\rA:=\{J^{ae|x}_\rA\}\in\sM^{\cA\times\cA}_\rA$ such that 
    \begin{align}
        \sum_{e\in\cA}J^{ae|x}_\rA&=F^{\circ a|x}_\rA\quad(\forall x\in\cX\setminus\{x^*\},\forall a\in\cA),\\
        \sum_{a\in\cA}J^{ae|x}_\rA&=F^{\circ e|x^*}_\rA\quad(\forall x\in\cX\setminus\{x^*\},\forall e\in\cA).
    \end{align}
    Construct a set of POVM $\{\G'^{\cA\times\cA|x}_\rA\}_{x\in\cX}$ where each element of a POVM $\G'^{\cA\times\cA|x}_\rA=\{G'^{ae|x}_\rA\}_{a,e\in\cA}$ is defined as
    \begin{equation}
        G'^{ae|x}_\rA:=
        \begin{dcases}
            J^{ae|x}_\rA\quad &(x\neq x^*),\\
            \delta_{ae}F^{\circ a|x^*}_\rA\quad &(x=x^*).
        \end{dcases}
    \end{equation}
    This is a feasible solution of Eq.~\eqref{eq:guessing_PMD} and achieves the optimal solution $p^{x^*}_\guess (\mbF^{\circ \cA|\cX}_\rA,\rho_\rA)=1$.
    
    Also, take POVMs $\{\G''^x_\rA\}_{x\in\cX}$ as a optimal solution of $p^{x^*}_\guess (\mbE^{\circ \cA|\cX}_\rA,\rho_\rA)$ where $\G''^x_\rA:=\{G''^{ae|x}_\rA\}\in\sM^{\cA\times\cA}_\rA$ such that 
    \begin{align}
        \sum_{e\in\cA}G''^{ae|x}_\rA&=E^{\circ a|x}_\rA\quad(\forall x\in\cX,\forall a\in\cA)),\\
        \sum_{a\in\cA}G''^{ae|x}_\rA&=\sum_{b\in\cA}G''^{be|x^*}_\rA\quad(\forall x\in\cX,\forall e\in\cA).
    \end{align}
    Now, define
    \begin{equation}
    \begin{split}
        G^{ae|x}_\rA&:=\{1-W^{x^*}(\mbM^{\cA|\cX}_\rA)\}G'^{ae|x}_\rA+W^{x^*}(\mbM^{\cA|\cX}_\rA)G''^{ae|x}_\rA\quad(\forall x\in\cX,\forall a,e\in\cA).
    \end{split}
    \end{equation}
    One can confirm that
    \begin{equation}
    \begin{split}
        \sum_{e\in\cA}G^{ae|x}_\rA&=\{1-W^{x^*}(\mbM^{\cA|\cX}_\rA)\}\sum_{e\in\cA}G'^{ae|x}_\rA+W^{x^*}(\mbM^{\cA|\cX}_\rA)\sum_{e\in\cA}G''^{ae|x}_\rA\\
        &=\{1-W^{x^*}(\mbM^{\cA|\cX}_\rA)\}F^{\circ a|x}_\rA+W^{x^*}(\mbM^{\cA|\cX}_\rA)E^{\circ a|x}_\rA\\
        &=M^{a|x}_\rA\quad(\forall x\in\cX,\forall a\in\cA),
    \end{split}
    \end{equation}
    and
    \begin{equation}
        \begin{split}
            \sum_{a\in\cA}G^{ae|x}_\rA&:=\{1-W^{x^*}(\mbM^{\cA|\cX}_\rA)\}\sum_{a\in\cA}G'^{ae|x}_\rA+W^{x^*}(\mbM^{\cA|\cX}_\rA)\sum_{a\in\cA}G''^{ae|x}_\rA\\
        &=\{1-W^{x^*}(\mbM^{\cA|\cX}_\rA)\}F^{\circ e|x^*}_\rA+W^{x^*}(\mbM^{\cA|\cX}_\rA)\sum_{b\in\cA}G''^{be|x^*}_\rA\quad(\forall x\in\cX,\forall e\in\cA).
        \end{split}
    \end{equation}
    Note that the right-hand side does not depend on $x$.
    Therefore, $G^{ae|x}_\rA$ is a feasible solution of \eqref{eq:guessing_PMD} that calculates $p^{x^*}_\guess (\mbM^{\cA|\cX}_\rA,\rho_\rA)$.
    Thus, we have
    \begin{equation}
    \begin{split}
        p^{x^*}_\guess (\mbM^{\cA|\cX}_\rA,\rho_\rA)&\ge \sum_{e\in\cA}\Tr[G^{ee|x^*}_\rA\rho_\rA]\\
        &=\{1-W^{x^*}(\mbM^{\cA|\cX}_\rA)\}\sum_{e\in\cA}\Tr[G'^{ee|x^*}_\rA\rho_\rA]+W^{x^*}(\mbM^{\cA|\cX}_\rA)\sum_{e\in\cA}\Tr[G''^{ee|x^*}_\rA\rho_\rA]\\
        &=1-W^{x^*}(\mbM^{\cA|\cX}_\rA)+W^{x^*}(\mbM^{\cA|\cX}_\rA)p^{x^*}_\guess (\mbE^{\circ\cA|\cX}_\rA,\rho_\rA)\\
        &\ge 1-W^{x^*}(\mbM^{\cA|\cX}_\rA)+\frac{W^{x^*}(\mbM^{\cA|\cX}_\rA)}{|\cA|}\\
        &=1-W^{x^*}(\mbM^{\cA|\cX}_\rA)\left(1-\frac{1}{|\cA|}\right).
    \end{split}
    \end{equation}
    The first inequality comes from the fact that $G^{ae|x}_\rA$ is a feasible solution but not necessarily optimal.
    The second inequality uses the fact that $p^{x^*}_\guess (\mbE^{\circ\cA|\cX}_\rA,\rho_\rA)$ is greater than or equal to $1/|\cA|$ because uniform random guessing is always possible.
    
    By combining this inequality with Proposition~\ref{proposition:assemblage_PMD}, we get
    \begin{equation}
        p_\guess^{x^*}(\sigma(\mbM^{\cA|\cX}_\rA,\rho_{\rA\rB}))\ge  p^{x^*}_\guess (\mbM^{\cA|\cX}_\rA,\rho_\rA)\ge 1-W^{x^*}(\mbM^{\cA|\cX}_\rA)\left(1-\frac{1}{|\cA|}\right),
    \end{equation}
    which completes the proof.
\end{proof}

\subsection{Tightness and the relation between star-incompatibility weight and guessing probability}

We mention that there is an example that achieves the equality in the inequality provided in Theorem~\ref{theorem:weight}.
\begin{example}
Let us consider the following set of POVMs $\mbM^{\cA|\cX}_\rA(w)$ on a qubit system $\cH_\rA$:
\begin{align}
    M^{0|0}_\rA(w)&:=(1-w)\one_\rA+w\ketbra{0}{0}_\rA,\quad M^{1|0}_\rA(w):=w\ketbra{1}{1}_\rA,\\
    M^{0|1}_\rA(w)&:=\ketbra{+}{+}_\rA,\quad M^{1|1}_\rA(w):=\ketbra{-}{-}_\rA,
\end{align}
where $\{\ket{0}_\rA,\ket{1}_\rA\}$ is a standard basis of $\cH_\rA$, and $\ket{+}_\rA:=\frac{1}{\sqrt{2}}(\ket{0}_\rA+\ket{1}_\rA)$ and $\ket{-}_\rA:=\frac{1}{\sqrt{2}}(\ket{0}_\rA-\ket{1}_\rA)$.
Then, it holds that
\begin{equation}
    W^0(\mbM^{\cA|\cX}_\rA(w))=w
\end{equation}
and 
\begin{equation}
    \min_{\rho_{\rA\rB}\in\cS(\cH_\rA\otimes\cH_\rB)}p^0_\guess(\sigma(\mbM^{\cA|\cX}_\rA(w),\rho_{\rA\rB}))=1-\frac{w}{2}.
\end{equation}
\end{example}

\begin{proof}

First, we show $W^0(\mbM^{\cA|\cX}_\rA(w))\le w$.
One can take a set of star-compatible POVMs $\mbF^{\cA|\cX}_\rA$ with elements
\begin{align}
    F^{0|0}_\rA&:=\one_\rA,\quad F^{1|0}_\rA:=\zero_\rA,\\
    F^{0|1}_\rA&:=\ketbra{+}{+}_\rA,\quad F^{1|1}_\rA:=\ketbra{-}{-}_\rA.
\end{align}
Also, take a set of POVMs $\mbE^{\cA|\cX}_\rA$ with elements
\begin{align}
    E^{0|0}_\rA&:=\ketbra{0}{0}_\rA,\quad E^{1|0}_\rA:=\ketbra{1}{1}_\rA,\\
    E^{0|1}_\rA&:=\ketbra{+}{+}_\rA,\quad E^{1|1}_\rA:=\ketbra{-}{-}_\rA.
\end{align}
Then, we have $M^{a|x}_\rA(w)=(1-w) F^{a|x}_\rA+wE^{a|x}_\rA$, so they are feasible POVMs.

Next, we show $W^0(\mbM^{\cA|\cX}_\rA(w))\ge w$.
Let the decomposition be $M^{a|x}_\rA(w)=(1-k) F^{a|x}_\rA+kE^{a|x}_\rA$.
Consider the case where $0<k<1$.
Since $M^{1|0}_\rA(w)=w\ketbra{1}{1}_\rA$, $E^{1|0}_\rA$ and $F^{1|0}_\rA$ must be proportional to $\ketbra{1}{1}_\rA$: $F^{1|0}_\rA=f\ketbra{1}{1}_\rA$ and $E^{1|0}_\rA=e\ketbra{1}{1}_\rA\;(0\le e,f\le 1)$.
Also, $F^{0|1}_\rA=E^{0|1}_\rA=\ketbra{+}{+}_\rA$ and $F^{1|1}_\rA=E^{1|1}_\rA=\ketbra{-}{-}_\rA$.
However, POVMs $\{\one_\rA-f\ketbra{1}{1}_\rA,f\ketbra{1}{1}_\rA\}$ and $\{\ketbra{+}{+}_\rA,\ketbra{-}{-}_\rA\}$ are incompatible if $f>0$~\cite{heinosaari2008notes}.
Therefore, $f$ must be $0$.
Then,
\begin{equation}
    M^{1|0}_\rA(w)=ke\ketbra{1}{1}=w\ketbra{1}{1}
\end{equation}
and thus, $k\ge w$ because $e\le 1$.
This completes the proof.
 
Since we consider the case where the dimension of $\rA$ and $\rB$ are the same, one can use Theorem~\ref{theorem:assemblage_PMD} and thus, 
\begin{equation}
    \min_{\rho_{\rA\rB}\in\cS(\cH_\rA\otimes\cH_\rB)}p^0_\guess(\sigma(\mbM^{\cA|\cX}_\rA(w),\rho_{\rA\rB}))=\min_{\rho_\rA\in\cS(\cH_\rA)}p^0_\guess(\mbM^{\cA|\cX}_\rA(w),\rho_\rA).
\end{equation}
Now, fix $\rho_\rA=\one_\rA/2$.
Then,
\begin{align}
    \Tr\left[M^{0|0}_\rA(w)\rho_\rA\right]&=1-w+\frac{w}{2}=1-\frac{w}{2},\\
    \Tr\left[M^{1|0}_\rA(w)\rho_\rA\right]&=\frac{w}{2}.
\end{align}
Since $0\le w\le 1$, Eve can take a strategy such that she always guesses that Alice's outcome is $0$, which makes the guessing probability $1-w/2$.
This means that the optimal guessing probability is larger than or equal to $1-w/2$.

In fact, however, $1-w/2$ is the optimal when we use $\rho_\rA=\one_\rA/2$.
We have the following equalities:
\begin{align}
    G^{00|0}_\rA+G^{01|0}_\rA&=(1-w)\one_\rA+w\ketbra{0}{0}_\rA,\\
    G^{10|0}+G^{11|0}&=w\ketbra{1}{1}_\rA,\\
    G^{00|1}_\rA+G^{01|1}_\rA&=\ketbra{+}{+}_\rA,\\
    G^{10|1}_\rA+G^{11|1}_\rA&=\ketbra{-}{-}_\rA,\\
    G^{00|0}_\rA+G^{10|0}_\rA&=G^{00|1}_\rA+G^{10|1}_\rA,\\
    G^{01|0}_\rA+G^{11|0}_\rA&=G^{01|1}_\rA+G^{11|1}_\rA.
\end{align}
By these constraints, one can set 
\begin{align}
    G^{11|0}_\rA&=\lambda\ketbra{1}{1}_\rA\quad(0\le\lambda\le w),\\
    G^{01|1}_\rA&=\alpha\ketbra{+}{+}_\rA\quad(0\le\alpha\le 1),\\
    G^{11|1}_\rA&=\beta\ketbra{-}{-}_\rA\quad(0\le\beta\le 1).
\end{align}

Note that we have
\begin{equation}
    G^{00|0}_\rA=(1-w)\one_\rA+w\ketbra{0}{0}_\rA-G^{01|0}_\rA=(1-w)\one_\rA+w\ketbra{0}{0}_\rA+G^{11|0}_\rA-(G^{01|1}_\rA+G^{11|1}_\rA).
\end{equation}
Then, it holds that
\begin{equation}
    G^{00|0}_\rA+G^{11|0}_\rA=(1-w)\one_\rA+w\ketbra{0}{0}_\rA+2G^{11|0}_\rA-(G^{01|1}_\rA+G^{11|1}_\rA).
\end{equation}
Also, since it holds that $G^{11|0}_\rA\preceq G^{01|1}_\rA+G^{11|1}_\rA$, one has $\bra{1} G^{11|0}_\rA\ket{1}_\rA\le\bra{1} G^{01|1}_\rA\ket{1}_\rA+\bra{1} G^{11|1}_\rA\ket{1}_\rA$, which leads to
\begin{equation}
    \lambda\le\frac{\alpha+\beta}{2}.
\end{equation}
By using this, one gets
\begin{equation}
    \begin{split}
        \Tr\left[(G^{00|0}_\rA+G^{11|0}_\rA)\frac{\one_\rA}{2}\right]=1-w+\frac{w}{2}+\lambda-\frac{\alpha+\beta}{2}\le1-\frac{w}{2}.
    \end{split}
\end{equation}
Therefore, we have
\begin{equation}
    p^0_\guess\left(\mbM^{\cA|\cX}_\rA(w),\frac{\one_\rA}{2}\right)=1-\frac{w}{2},
\end{equation}
which achieves the equality of the inequality in Theorem~\ref{theorem:weight}.

\end{proof}

\begin{example}
    Let us consider the following set of POVMs $\mbN^{\cA|\cX}_\rA(w)$ on a qubit system $\cH_\rA$:
    \begin{align}
        N^{0|0}_\rA(w)&=\ketbra{0}{0}_\rA,\quad N^{1|0}_\rA(w)=\ketbra{1}{1}_\rA,\\
        N^{0|1}_\rA(w)&=\frac{1}{2}\left(\one_\rA+wX_\rA\right),\quad N^{1|1}_\rA(w)=\frac{1}{2}\left(\one_\rA-wX_\rA\right)\quad(0\le w\le 1),
    \end{align}
    where $\{\ket{0}_\rA,\ket{1}_\rA\}$ is a standard basis of $\cH_\rA$.
    Then, it holds that
    \begin{equation}
        W^0(\mbN^{\cA|\cX}_\rA(w))=w
    \end{equation}
    and
    \begin{equation}
        \min_{\rho_{\rA\rB}\in\cS(\cH_\rA\otimes\cH_\rB)}p^0_\guess(\sigma(\mbN^{\cA|\cX}_\rA(w),\rho_{\rA\rB}))=\frac{1+\sqrt{1-w^2}}{2}
    \end{equation}
\end{example}

\begin{proof}
First, we prove $W^0(\mbN^{\cA|\cX}_\rA(w))\le w$.
This is because one can take a star-compatible set of POVMs
\begin{align}
    F^{0|0}_\rA&:=\ketbra{0}{0}_\rA,\quad F^{1|0}_\rA:=\ketbra{1}{1}_\rA,\\
    F^{0|1}_\rA&:=\frac{1}{2}\one_\rA,\quad F^{1|1}_\rA:=\frac{1}{2}\one_\rA.
\end{align}
and take a set of POVMs $\mbE^{\cA|\cX}_\rA$ with elements
\begin{align}
    E^{0|0}_\rA&:=\ketbra{0}{0}_\rA,\quad E^{1|0}_\rA:=\ketbra{1}{1}_\rA,\\
    E^{0|1}_\rA&:=\frac{1}{2}\left(\one_\rA+X_\rA\right)\quad E^{1|1}_\rA:=\frac{1}{2}\left(\one_\rA-X_\rA\right).
\end{align}
so that $N^{a|x}_\rA=(1-w)F^{a|x}_\rA+wE^{a|x}_\rA$, which implies $W^0(\mbN^{\cA|\cX}_\rA(w))\le w$.

We also show $W^0(\mbN^{\cA|\cX}_\rA(w))\ge w$.
Since $N^{0|0}_\rA(w)$ and $N^{1|0}_\rA(w)$ are rank one, taking $F^{0|0}_\rA:=\ketbra{0}{0}_\rA$ and $F^{1|0}_\rA:=\ketbra{1}{1}_\rA$ is the only option.
Thus, $F^{0|1}_\rA$ and $F^{1|1}_\rA$ must be diagonalized by the $Z$-basis:
\begin{align}
    F^{0|1}_\rA&=u\ketbra{0}{0}_\rA+v\ketbra{1}{1}_\rA,\\
    F^{1|1}_\rA&=\one_\rA-F^{0|1}_\rA.
\end{align}
Therefore,
\begin{align}
    (1-k)\ketbra{0}{0}_\rA+k E^{0|0}_\rA&=\ketbra{0}{0}_\rA,\\
    (1-k)\ketbra{1}{1}_\rA+k E^{1|0}_\rA&=\ketbra{1}{1}_\rA,\\
    (1-k)(u\ketbra{0}{0}_\rA+v\ketbra{1}{1}_\rA)+kE^{0|1}_\rA&=\frac{1}{2}\left(\one_\rA+wX_\rA\right),\\
    (1-k)(\one_\rA-(u\ketbra{0}{0}_\rA+v\ketbra{1}{1}_\rA))+kE^{1|1}_\rA&=\frac{1}{2}\left(\one_\rA-wX_\rA\right).
\end{align}
Let us focus on the third constraint.
It holds that
\begin{equation}
    \begin{split}
        \bra{0}_\rA\{(1-k)(u\ketbra{0}{0}_\rA+v\ketbra{1}{1}_\rA)+kE^{0|1}_\rA\ket{1}_\rA&=k\bra{0}_\rA E^{0|1}_\rA\ket{1}_\rA\\
        &=\bra{0}_\rA\left\{\frac{1}{2}\left(\one_\rA+wX_\rA\right)\right\}\ket{1}_\rA\\
        &=\frac{w}{2}.
    \end{split}
\end{equation}

Note that the absolute value of the off-diagonal elements of a two-valued POVM on a two-dimensional Hilbert space is not larger than $1/2$.
Let us take a two-valued POVM $\{E_\rA,\one_\rA-E_\rA\}$ on $\cH_\rA$ with $\dim\cH_\rA=2$. 
According to Ref.~\cite{kiukas2022amount}, it holds that
\begin{equation}
    |\bra{0}_\rA E_\rA\ket{1}_\rA|+|\bra{0}_\rA (\one_\rA-E_\rA)\ket{1}_\rA|=2|\bra{0}_\rA E_\rA\ket{1}_\rA|\le 1.
\end{equation}
Therefore, $|\bra{0}_\rA E_\rA\ket{1}_\rA|\le 1/2$.
Therefore, one has
\begin{equation}
    \frac{w}{2}=k\bra{0}_\rA E^{0|1}_\rA\ket{1}_\rA \le \frac{k}{2},
\end{equation}
which means that $k\ge w$.
This completes the proof of $W^0(\mbN^{\cA|\cX}_\rA(w))=w$.

Next, we show
\begin{equation}
    p^0_\guess\left(\mbN^{\cA|\cX}_\rA(w),\frac{\one_\rA}{2}\right)=\frac{1+\sqrt{1-w^2}}{2}.
\end{equation}
Since $N^{0|0}_\rA(w)$ and $N^{1|0}_\rA(w)$ are rank one, we have
\begin{align}
    G^{00|0}_\rA(w)&=\alpha\ketbra{0}{0}_\rA,\\
    G^{01|0}_\rA(w)&=(1-\alpha)\ketbra{0}{0}_\rA,\\
    G^{10|0}_\rA(w)&=\beta\ketbra{1}{1}_\rA\\
    G^{11|0}_\rA(w)&=(1-\beta)\ketbra{1}{1}_\rA,\\
    G^{00|1}_\rA(w)+G^{01|1}_\rA(w)&=\frac{1}{2}(\one_\rA+wX_\rA)\\
    G^{10|1}_\rA(w)+G^{11|1}_\rA(w)&=\frac{1}{2}(\one_\rA-wX_\rA)
\end{align}
We also have
\begin{align}
    G^{00|0}_\rA(w)+G^{10|0}_\rA(w)&=\alpha\ketbra{0}{0}_\rA+\beta\ketbra{1}{1}_\rA=G^{00|1}_\rA(w)+G^{10|1}_\rA(w)\\
    G^{01|0}_\rA(w)+G^{11|0}_\rA(w)&=(1-\alpha)\ketbra{0}{0}_\rA+(1-\beta)\ketbra{1}{1}_\rA=G^{01|1}_\rA(w)+G^{11|1}_\rA(w),
\end{align}
which means that $\{\alpha\ketbra{0}{0}_\rA+\beta\ketbra{1}{1}_\rA,\one_\rA-(\alpha\ketbra{0}{0}_\rA+\beta\ketbra{1}{1}_\rA)\}$ and $\{\frac{1}{2}(\one_\rA+wX_\rA),\frac{1}{2}(\one_\rA-wX_\rA)\}$ are compatible.
From the above condition, we get
\begin{align}
    G^{00|1}_\rA-G^{11|1}_\rA&=\alpha\ketbra{0}{0}_\rA+\beta\ketbra{1}{1}_\rA-\frac{1}{2}(\one_\rA-wX_\rA)=
    \begin{pmatrix}
        \alpha-1/2 & w/2\\
        w/2 & \beta-1/2
    \end{pmatrix}
\end{align}
and
\begin{align}
    G^{10|1}_\rA-G^{01|1}_\rA&=\alpha\ketbra{0}{0}_\rA+\beta\ketbra{1}{1}_\rA-\frac{1}{2}(\one_\rA+wX_\rA)=
    \begin{pmatrix}
        \alpha-1/2 & -w/2\\
        -w/2 & \beta-1/2
    \end{pmatrix}.
\end{align}
Then, the two eigenvalues are
\begin{equation}
    \lambda_\pm(G^{00|1}_\rA-G^{11|1}_\rA)=\frac{1}{2}\left(\alpha+\beta-1\pm\sqrt{(\alpha-\beta)^2+w^2}\right)
\end{equation}
and thus, it holds that
\begin{equation}
    \|G^{00|1}_\rA-G^{11|1}_\rA\|_1=\max\left\{|\alpha+\beta-1|,\sqrt{(\alpha-\beta)^2+w^2}\right\}
\end{equation}
The same calculation also leads to
\begin{equation}
    \|G^{10|1}_\rA-G^{01|1}_\rA\|_1=\max\left\{|\alpha+\beta-1|,\sqrt{(\alpha-\beta)^2+w^2}\right\}.
\end{equation}
Note that since it holds that
\begin{equation}
\begin{split}
    \|G^{00|1}_\rA-G^{11|1}_\rA\|_1+\|G^{10|1}_\rA-G^{01|1}_\rA\|_1&\le\|G^{00|1}_\rA\|_1+\|G^{11|1}_\rA\|_1+\|G^{10|1}_\rA\|_1+\|G^{01|1}_\rA\|_1\\
    &=\Tr[G^{00|1}_\rA]+\Tr[G^{11|1}_\rA]+\Tr[G^{10|1}_\rA]+\Tr[G^{01|1}_\rA]\\
    &=\Tr[G^{00|1}_\rA+G^{11|1}_\rA+G^{10|1}_\rA+G^{01|1}_\rA]\\
    &=2,
\end{split}
\end{equation}
one must have
\begin{equation}
    2\sqrt{(\alpha-\beta)^2+w^2}\le 2,
\end{equation}
which implies that
\begin{equation}
    w^2+(\alpha-\beta)^2\le 1.
\end{equation}
Now, the guessing probability for this $\{G^{ae|x}_\rA\}$ is
\begin{equation}
    \Tr\left[\left(G^{00|0}_\rA+G^{11|0}_\rA\right)\frac{I_\rA}{2}\right]=\frac{1+\alpha-\beta}{2}\le\frac{1+\sqrt{1-w^2}}{2}.
\end{equation}
This equality is achievable when $\alpha-\beta=\sqrt{1-w^2}$, so we have $p^0_\guess\left(\mbN^{\cA|\cX}_\rA(w),\frac{\one_\rA}{2}\right)=\frac{1+\sqrt{1-w^2}}{2}$, which means that
\begin{equation}
    \min_{\rho_\rA\in\cS(\cH_\rA)}p^0_\guess(\mbN^{\cA|\cX}_\rA,\rho_{\rA})\le\frac{1+\sqrt{1-w^2}}{2}.
\end{equation}

We next show the converse inequality.
Introduce $\gamma:=\frac{1+\sqrt{1-w^2}}{2}$ and take
  \begin{align}
    G^{00|0}_\rA(w)&=\gamma\ketbra{0}{0}_\rA,\\
    G^{01|0}_\rA(w)&=(1-\gamma)\ketbra{0}{0}_\rA,\\
    G^{10|0}_\rA(w)&=(1-\gamma)\ketbra{1}{1}_\rA\\
    G^{11|0}_\rA(w)&=\gamma\ketbra{1}{1}_\rA.
\end{align}
We have
\begin{align}
    G^{00|0}_\rA+G^{10|0}_\rA&=\frac{1}{2}(\one_\rA+(2\gamma-1)Z_\rA)\\
    G^{01|0}_\rA+G^{11|0}_\rA&=\frac{1}{2}(\one_\rA-(2\gamma-1)Z_\rA)
\end{align}
Now, take
\begin{equation}
    G^{a_1e|1}_\rA=\frac{1}{4}(\one_\rA+(-1)^{a_1} wX_\rA+(-1)^e(2\gamma-1)Z_\rA),
\end{equation}
which satisfies
\begin{align}
    \sum_e G^{ae|1}_\rA&=\frac{1}{2}(\one_\rA+(-1)^awX_\rA)=N^{a_1|1}_\rA\\
    \sum_{a_1}G^{a_1e|1}_\rA&=\frac{1}{2}(\one_\rA\pm(2\gamma-1)Z_\rA)
\end{align}
Therefore, $\{G^{ae|x}_\rA\}$ is a feasible solution.
Finally, it holds that
\begin{equation}
    \sum_{e}\Tr[G^{ee|0}_\rA\rho_\rA]=\Tr[\gamma\one_\rA\rho_\rA]=\frac{1+\sqrt{1-w^2}}{2}
\end{equation}
for all state $\rho_\rA$.
Therefore,
\begin{equation}
    \min_{\rho_\rA\in\cS(\cH_\rA)}p^0_\guess(\mbN^{\cA|\cX}_\rA,\rho_{\rA})\ge\frac{1+\sqrt{1-w^2}}{2}.
\end{equation}

\end{proof}

\end{document}